\documentclass[reprint,aps,amsmath,amssymb,pra,superscriptaddress,floatfix,10pt,tightenlines]{revtex4-2}
\usepackage{amsmath, amssymb, amsthm}
\usepackage{mathtools}
\usepackage{color}
\usepackage{braket}
\usepackage{comment}
\usepackage{bbm}
\usepackage{bm}
\usepackage[breaklinks=true]{hyperref}
\usepackage{url}
\usepackage{color}

\DeclareMathOperator{\tr}{Tr}

\DeclareMathOperator{\entropy}{H}
\DeclareMathOperator{\relent}{D}
\DeclareMathSymbol{\mhyph}{\mathalpha}{operators}{`-}

\newcommand{\cc}{\mathrm{c}\mhyph\mathrm{c}}
\newcommand{\cq}{\mathrm{c}\mhyph\mathrm{q}}
\newcommand{\qc}{\mathrm{q}\mhyph\mathrm{c}}
\newcommand{\I}{\mathbbm{1}}

\allowdisplaybreaks[4]

\theoremstyle{plain}
\newtheorem{theorem}{Theorem}
\newtheorem{lemma}[theorem]{Lemma}

\newtheorem{proposition}[theorem]{Proposition}
\theoremstyle{definition}
\newtheorem{definition}[theorem]{Definition}

\theoremstyle{remark}
\newtheorem{conjecture}{Conjecture}
\newtheorem{remark}[conjecture]{Remark}

\begin{document}
\author{Kohdai Kuroiwa}
\email{kkuroiwa@uwaterloo.ca}
\affiliation{Institute for Quantum Computing and Department of Physics and Astronomy, University of Waterloo, 200 University Avenue West, Waterloo, Ontario, Canada, N2L 3G1}
\author{Hayata Yamasaki}
\email{hayata.yamasaki@gmail.com}
\affiliation{Institute for Quantum Optics and Quantum Information (IQOQI), Austrian Academy of Sciences, Boltzmanngasse 3, 1090 Vienna, Austria}
\affiliation{Institute for Atomic and Subatomic Physics, Vienna University of Technology, 1020 Vienna, Austria}

\date{\today}
\title{Asymptotically Consistent Measures of General Quantum Resources: Discord, Non-Markovianity, and Non-Gaussianity}

\begin{abstract} 
    Quantum resource theories provide a unified framework to quantitatively analyze inherent quantum properties as resources for quantum information processing. 
    So as to investigate the best way for quantifying resources, desirable axioms for resource quantification have been extensively studied through axiomatic approaches. 
    However, a conventional way of resource quantification by resource measures with such desired axioms may contradict rates of asymptotic transformation between resourceful quantum states due to an approximation in the transformation.
    In this paper, we establish an alternative axiom, \textit{asymptotic consistency} of resource measures, and we investigate \textit{asymptotically consistent resource measures}, which quantify resources without contradicting the rates of the asymptotic resource transformation. 
    We prove that relative entropic measures are consistent with the rates for a broad class of resources, \textit{i.e.}, all convex finite-dimensional resources, 
    \textit{e.g.}, entanglement, coherence, and magic, 
    and even some nonconvex or infinite-dimensional resources such as quantum discord, non-Markovianity, and non-Gaussianity. 
    These results show that consistent resource measures are widely applicable to the quantitative analysis of various inherent quantum-mechanical properties. 
\end{abstract}

\maketitle

\section{Introduction}
Quantum resource theories (QRTs)~\cite{Chitambar2018,Kuroiwa2020} provide a unified framework for quantitatively analyzing quantum properties, such as entanglement~\cite{Horodecki2009} and magic~\cite{Veitch2014, Howard2017,liu2020manybody}, 
which underlies the advantage of quantum information processing over classical information processing. 
General frameworks to reveal the universal properties of quantum resources have been widely studied~\cite{Horodecki2013b, Brandao2015, Korzekwa2019,Liu_CL2019,Vijayan2019,Liu_ZW2019a,Takagi2019b,Takagi2019a,Fang2020,Regula2020,Sparaciari2020,Kuroiwa2020}, 
and quantification of resources is one of the major focuses of QRTs~\cite{Liu2017,Regula2017,Bromley2018,Anshu2018,Uola2019,gonda2019monotones,Regula2021,Lami2021}.
To quantify resources, we use real-valued functions of states called \textit{resource measures}. 
Resource measures can quantify resources without contradicting one-shot convertibility of resources; that is, the resource amount quantified by a resource measure does not increase when we transform a resource by free operations without error.
For example, \textit{the relative entropy of resource}, that is, the relative entropy between a given state and its closest free state, is conventionally used in various QRTs~\cite{Vedral1997,Veitch2014,Baumgratz2014,Horodecki2013,Modi2010,Ibinson2008,Geroni2008,lami2021attainability}.
With various resource measures proposed, desirable properties of resource measures have been extensively studied through axiomatic approaches~\cite{Guifre2000,Horodecki2001,Donald2002,Synak2006,Girard2014,Regula2017,Kuroiwa2020} to seek the best resource measure. 

On the other hand, it is known that resources are not necessarily comparable in terms of the exact convertibility.
For example, in the QRT of magic for qutrits, there are two classes of states impossible to be exactly converted to each other using free operations~\cite{Veitch2014}.  
Moreover, resource measures with conventionally adopted axioms, \textit{i.e.,} asymptotic continuity and additivity, may contradict rates of asymptotic state conversion~\cite{Kuroiwa2020}, where many copies of a given state are converted by free operations into many copies of a target state within a vanishingly small but nonzero error. 
Therefore, it is vital to investigate properties of resource measures associated with the asymptotic state conversion in addition to the exact conversion. 
While these two concepts, resource measures and asymptotic state conversion, were previously discussed separately, 
we proposed in Ref.~\cite{Kuroiwa2020} a concept of \textit{asymptotically consistent resource measures}, or \textit{consistent resource measures} for short, which quantifies resources without contradicting the rates of the asymptotic state conversion as well as the exact conversion. 
Reference~\cite{Kuroiwa2020} also shows that a relative entropic measure called \textit{the regularized relative entropy of resource} serves as a consistent resource measure in QRTs with particular restrictions, namely, convex and finite-dimensional QRTs with a full-rank free state.

However, physically well-motivated resources do not necessarily satisfy these restrictions. 
For example, the sets of states that have no quantum discord, \textit{e.g.}, classical-classical states and classical-quantum states, are not convex~\cite{Bera2017}.
In addition, the set of quantum Markov chains~\cite{Accardi1983} and that of Gaussian states~\cite{Weedbrook2012} are also known to be nonconvex.
Moreover, the Gaussian states are defined on an infinite-dimensional state space.
Therefore, when we regard these properties as quantum resources, the existing technique for proving the consistency of the entropic measure is no longer directly applicable~\cite{Kuroiwa2020,Weis2012}, 
and a substantial breakthrough is needed for further development of general QRTs that cover such nonconvex or infinite-dimensional resources.

In this paper, we investigate the regularized relative entropy of resource as a prospective candidate for a consistent resource measure even for the general classes of resources.
We prove the consistency of the regularized relative entropy of resource in \textit{all} finite-dimensional convex QRTs even without full-rank free states. 
Moreover, as physically well-motivated nonconvex or infinite-dimensional resources, we study discord, non-Markovianity, and non-Gaussianity. 
Even though the proof technique of the previous work cannot be straightforwardly applied due to the nonconvexity and the infinite-dimensionality, 
we analyze the consistency of the regularized relative entropy of resource for these three resources. 
A more detailed overview of our results is at the beginning of Sec.~\ref{sec:results}.
Here, we emphasize that the concept of consistent resource measures accomplishes the fundamental and physically intuitive property that resources do not increase under free operations in both asymptotic and one-shot state conversions.
Even though the regularized relative entropy of resource is a measure asymptotically defined, which may not be extracted with one-shot manipulation, 
the regularized relative entropy still qualifies as a valid resource measure with monotonicity under one-shot state conversion. 

Our results show the existence of a consistent resource measure even for physically important classes of nonconvex or infinite-dimensional resources, namely, discord, non-Markovianity, and non-Gaussianity,
which are not covered in our initial work~\cite{Kuroiwa2020} introducing consistent resource measures. 
Our analysis should be an essential first step for the breakthrough in general QRTs to broadly cover the general class of resources.
We believe that our investigation for consistent resource measures leads to further understandings of quantifications of resources useful for studying a larger class of physically well-motivated quantum phenomena. 

\section{Consistent resource measure}~\label{sec:prelim}
In this section, we provide a brief review of consistent resource measures. 
For more details, see Sections 6.3 and 6.4 of Ref.~\cite{Kuroiwa2020}.
Throughout this paper, we let $\mathcal{D}(\mathcal{H})$ denote the set of states on a quantum system $\mathcal{H}$, and we consider a subset $\mathcal{S}(\mathcal{H}) \subseteq \mathcal{D}(\mathcal{H})$ to be the set of states of interest in QRTs.
In a QRT, free operations are defined as a subclass of quantum operations (linear completely positive and trace-preserving maps) that is closed under composition and tensor product, and includes the identity map and the (partial) trace~\cite{Chitambar2018,Kuroiwa2020}. 
Free states on $\mathcal{H}$ are states in $\mathcal{S}(\mathcal{H})$ into which an arbitrary state can be converted by free operations.
To begin with, we provide a definition of a resource measure. 
\begin{definition}[Resource Measure]
    A resource measure $R_\mathcal{H}$ is a family of real functions from $\mathcal{S}\left(\mathcal{H}\right)$ to $\mathbb{R}$ satisfying \textit{monotonicity}:
    for any states $\phi \in \mathcal{S}(\mathcal{H})$ and $\psi \in \mathcal{S}(\mathcal{H}')$, 
    \begin{equation*}
        \exists \mathcal{N}: \mathrm{free\,\, operation \,\, s.t.\,\,} \mathcal{N}(\phi) = \psi \Rightarrow R_\mathcal{H}(\phi) \geqq R_{\mathcal{H}'}(\psi).
    \end{equation*} 
    We may omit the subscript of $R_\mathcal{H}$ to write $R$ for brevity when the system $\mathcal{H}$ is obvious.
\end{definition}

In Ref.~\cite{Kuroiwa2020}, a concept of \textit{consistency} of a resource measure is introduced, which we investigate in this paper. 
As implied in the definition, a resource measure $R$ quantifies resources consistently with the exact state conversion under the free operations. 
By also considering the consistency with the rate of the asymptotic state conversion, we give the definition of a consistent resource measure.
\begin{definition}[Asymptotically Consistent Resource Measure~\cite{Kuroiwa2020}]~\label{def_consistency}
    For quantum systems $\mathcal{H}$ and $\mathcal{H}'$, a resource measure $R$ is called \textit{asymptotically consistent}, or \textit{consistent} for short, if for any states $\phi \in \mathcal{S}\left(\mathcal{H}\right)$ and $\psi \in \mathcal{S}\left(\mathcal{H}'\right)$, it holds that 
    \begin{equation}\label{ineq_consistency}
        R_{\mathcal{H}'}\left(\psi\right)r\left(\phi \to \psi \right) \leqq R_{\mathcal{H}}\left(\phi\right).
    \end{equation} 
    In Eq.~\eqref{ineq_consistency}, $r\left(\phi \to \psi \right)$ is the rate of the asymptotic state conversion between $\phi$ and $\psi$ definied as 
    \begin{equation*}
        \begin{aligned}
        r\left(\phi\to\psi\right)\coloneqq
        &\inf \Big\{r\geqq 0:\exists\left(\mathcal{N}_n:\mathrm{free\,\, operation}\right)_{n\in\mathbb{N}},\\
        &\quad\liminf_{n\to\infty}\left\|\mathcal{N}_n\left(\phi^{\otimes n}\right)-\psi^{\otimes \left\lceil rn\right\rceil}\right\|_1 = 0\Big\}, 
        \end{aligned}
    \end{equation*}
    where  $\|\cdot\|_1$ is the trace norm, and $\lceil{}\cdots{}\rceil$ is the ceiling function.
\end{definition}
We here note that from the definition~\eqref{ineq_consistency}, consistent resource measures are weakly additive~\cite{Kuroiwa2020}; 
that is, $R(\phi^{\otimes n}) = nR(\phi)$ for all states $\phi$ and all positive integers $n$.
In Ref.~\cite{Kuroiwa2020}, a sufficient condition for the \textit{regularization} $R^{\infty}$ of a resource measure $R$, defined as $R^\infty(\phi) \coloneqq \lim_{n\to\infty} R(\phi^{\otimes n})/n$, to be consistent was provided.
The regularization of a resource measure also serves as a resource measure.
The sufficient condition consists of the following conventionally considered properties for a resource measure:
\begin{enumerate}
    \item \textit{Asymptotic continuity}: 
        For any sequence of positive integers ${\left(n_i\right)}_{i\in\mathbb{N}}$, and any sequences of states ${\left(\phi_{n_i} \in \mathcal{S}\left(\mathcal{H}^{\otimes n_i}\right)\right)}_i$ and ${\left(\psi_{n_i} \in \mathcal{S}\left(\mathcal{H}^{\otimes n_i}\right)\right)}_i$ satisfying $\lim_{i \to \infty} \left\|\phi_{n_i}-\psi_{n_i}\right\|_1 =0$, 
        it holds that $\lim_{i\to\infty}|R_{\mathcal{H}^{\otimes n_i}}(\phi_{n_i}) - R_{\mathcal{H}^{\otimes n_i}}(\psi_{n_i})|/n_i = 0$. 

    \item \textit{Subadditivity}: 
        For any states $\phi$ and $\psi$, it holds that 
        $R\left(\phi\otimes\psi\right)\leqq R\left(\phi\right)+R\left(\psi\right)$.
\end{enumerate}
Note that, in Ref.~\cite{ferrari2021asymptotic}, another sufficient condition for consistency of resource measures is given; 
in this paper, we investigate the sufficient condition shown in the following lemma~\cite{Kuroiwa2020}.
\begin{lemma}[Sufficient Condition for Consistency]~\label{lemma_sufficient_condition}
    The regularization $R^\infty$ of a resource measure $R$ is consistent if $R$ is asymptotically continuous and subadditive. 
\end{lemma}
From this lemma, a promising way for constructing a consistent resource measure is to consider a subadditive resource measure and then to check its asymptotic continuity.
As a subadditive resource measure, we here consider the relative entropy of resource $R_\textup{R}$, which is widely studied in various QRTs~\cite{Vedral1997,Veitch2014,Baumgratz2014,Horodecki2013,Modi2010,Ibinson2008,Geroni2008}. 
For our purpose, it may not be suitable to use non-subadditive measures such as robustness-based measures~\cite{Vidal1999,Steiner2003,Piani2016,Howard2017,Regula2017}.
\begin{definition}[Relative Entropy of Resource]
The relative entropy of resource $R_\textup{R}$ is defined as
\begin{equation}\label{def:relint}
    R_\textup{R}(\phi) \coloneqq \min_{\psi:\mathrm{free}} \relent(\phi\|\psi)
\end{equation} 
for all states $\phi$, where $\relent(\cdot\|\cdot)$ is the relative entropy $\relent(\rho||\sigma)\coloneqq \tr(\rho\log_2\rho - \rho\log_2\sigma)$. 
\end{definition}
The regularization of the relative entropy of resource $R^{\infty}_R$ is called \textit{the regularizaed relative entropy of resource}.
Reference~\cite{brandao2010} gives an operational interpretation of the regularlized relative entropy of resource; 
the regularized relative entropy of resource is considered as the optimal rate of asymptotic hypothesis testing where one aims to distinguish a given state and its closest free state.

\section{Main results}~\label{sec:results}
In this section, we investigate the consistency of the regularized relative entropy of resource for general resources of physical importance. 
Before we go through further details, we briefly overview the backgrounds and our contributions. 

In Ref.~\cite{Kuroiwa2020}, it was shown that the regularized relative entropy of resource is consistent in all finite-dimensional convex QRTs with a full-rank free state. 
The proof was based on asymptotic continuity of the relative entropy of resource in such QRTs~\cite{Synak2006,Winter2016}.

Here, we consider relaxing the assumption. 
First, in Sec.~\ref{subsec:finite_convex}, we study the necessity of full-rank free states. 
Indeed, by refining the definition of consistency, we show that the regularized relative entropy of resource is consistent in finite-dimensional convex QRTs even without full-rank free states.
Next, we investigate physically well-motivated nonconvex or infinite-dimensional resources. 
We will overview the physical significance of these resources at the beginning of each section. 
For nonconvex resources, namely, discord (Sec.~\ref{subsec:discord}) and non-Markovianity (Sec.~\ref{subsec:non-Markovianity}), 
despite the existence of a counterexample for asymptotic continuity of the relative entropy of resource~\cite{Weis2012}, 
we prove that the regularized relative entropy of resource is consistent. 
Moreover, we analyze an infinite-dimensional resource, non-Gaussianity, in Sec.~\ref{subsec:non-Gaussianity}. 
Contrary to the statement in the previous work~\cite{Geroni2008}, we prove that the relative entropy of non-Gaussianity is not continuous even with a reasonable energy constraint. 
On the other hand, if we take the convex hull of the set of Gaussian states, we show that the relative entropy of non-Gaussianity is asymptotically continuous under an appropriate energy constraint; 
therefore, the regularized relative entropy can be employed as a consistent resource measure under the energy constraint in this convex but infinite-dimensional QRT. 

\subsection{Finite dimensional convex QRTs}~\label{subsec:finite_convex}
In this section, we show that the regularized relative entropy of resource is indeed consistent in all finite-dimensional convex QRTs even without full-rank free states. 
This result contributes to removing an arguably artificial restriction of imposing full-rankness of free states in considering general QRTs.
Here, the problem is that if no full-rank free state exists, the regularized relative entropy of resource $R_\textup{R}$ in Eq.~\eqref{def:relint} may become infinite, and the definition of consistency may become an indeterminate form. 
To resolve these problems, we consider the inequality~\eqref{ineq_consistency} to be also true 
    when $R(\phi) = \infty$ or
    when $R(\phi) < \infty$, $R(\psi) = \infty$, and $r(\phi\to\psi) = 0$. 
If the set of free states contains a full-rank state, the regularized relative entropy of resource becomes consistent as shown in the previous work. 
As shown in the Supplemental Materials, we show that these two cases cover all the cases where $R$ diverges to infinity. 
Therefore, under this refined definition to take infinities into account, 
the regularized relative entropy of resources can be consistent even in the QRTs without full-rank free states.

\subsection{Discord}~\label{subsec:discord}
In this section, we investigate consistent resource measures for discord. 
Discord is conceptually understood as a form of quantum correlation; it is defined as the difference between total (quantum and classical) correlation and classical correlation~\cite{Bera2017}. 
While entanglement is widely appreciated as a quantum correlation enhancing various quantum information processing tasks~\cite{Horodecki2009}, 
it has been revealed that discord serves as a resource for several tasks even without the existence of entanglement~\cite{Datta2008, Dakic2012, Li2012, MADHOK2013, Gu2012, Pirandola2014, Weedbrook2016}. 
In addition, discord also has operational meanings in several tasks~\cite{Madhok2011, Cavalcanti2011}. 
From such operational significance of discord, it is essential to establish a consistent way for quantitative analysis of discord. 

To analyze discord for this purpose, one may consider a QRT of discord by choosing, as free operations, a class of operations that preserves a set of states with zero discord. 
As explained in Sec.~\ref{sec:prelim}, each choice of free operations determines the corresponding set of free states. 
For discord, there have been three kinds of sets of states conventionally considered as free states.
Here, we only consider the bipartite case for simplicity of presentation, while it is also possible to consider multipartite cases~\cite{Modi2010}.  
Hereafter, we consider two parties $A$ and $B$ with finite-dimensional quantum systems $\mathcal{H}^{A}$ and $\mathcal{H}^{B}$ respectively. 
Then, let $D$ denote the dimension of the composite system $\mathcal{H}^{A}\otimes \mathcal{H}^{B}$.
The first set of free states is the set of classical-classical states
$\cc \coloneqq \left\{\sum_{k}p_{k}\ket{a_k}\bra{a_k}^{A}\otimes\ket{b_k}\bra{b_k}^{B}\right\}$, 
where $\{\ket{a_k}^{A}\}$ is an orthonormal basis of $\mathcal{H}^{A}$ and $\{\ket{b_k}^{B}\}$ is an orthonormal basis of $\mathcal{H}^{B}$. 
Note that these bases are not necessarily fixed to some standard bases in contrast to the QRT of coherence~\cite{Streltsov2017}.
The second one is the set of quantum-classical states
$\qc \coloneqq \left\{\sum_{k}p_{k}\rho_k^{A}\otimes\ket{b_k}\bra{b_k}^{B}\right\}$,
and the last one is the set of classical-quantum states $\cq$, where the roles of $A$ and $B$ in 
$\qc$ are interchanged.
Due to the symmetry of the definitions of $\qc$ and $\cq$, we only consider $\cc$ and $\qc$ in the rest of this section.
Since the bases $\{\ket{a_k}^{A}\}$ and $\{\ket{b_k}^{B}\}$ are not fixed but arbitrary, it is clear that none of these three sets are convex. 
Then, for each of these sets of free states, discord can be analyzed in the framework of QRTs. 
Due to the nonconvexity of these sets, however, these QRTs of discord are nonconvex QRTs.

Based on these free states, we can define the relative entropy of discord as
\begin{equation}
    D^{\mathcal{{F}}}_{\textup{rel}}(\phi^{AB}) \coloneqq \min_{\psi^{AB} \in \mathcal{F}} \relent(\phi^{AB}\|\psi^{AB}),
\end{equation}
where $\mathcal{F}$ is either $\cc$ or $\qc$. 
Indeed, we derive a simple expression of the relative entropic measure for any choice of free states. 
In more detail, for all states $\phi^{AB} \in \mathcal{D}(\mathcal{H}^{A}\otimes\mathcal{H}^{B})$, it holds that
\begin{equation}~\label{eq:relent_discord}
    D^{\mathcal{F}}_{\textup{rel}}(\phi^{AB}) = \min_{\chi^{\mathcal{F}}_{\phi}} \entropy(\chi^{\mathcal{F}}_{\phi}) - \entropy(\phi^{AB}),
\end{equation} 
with the von Neumann entropy $\entropy(\rho)\coloneqq -\tr(\rho \log_2 \rho)$. 
The minimization in Eq.~\eqref{eq:relent_discord} is taken over the free states 
$\chi^{\cc}_{\phi} \coloneqq \sum_{k} (\bra{a_k,b_k})\phi^{AB}(\ket{a_k,b_k})\ket{a_k,b_k}\bra{a_k,b_k}^{AB}$ or 
$\chi^{\qc}_{\phi} \coloneqq \sum_{k} (\mathbbm{1}^{A}\otimes\bra{b_k})\phi^{AB}(\mathbbm{1}^{A}\otimes\ket{b_k}) \otimes\ket{b_k}\bra{b_k}^{B}$
with the identity operator $\mathbbm{1}$. 
See also Ref.~\cite{Modi2010} for derivation of this expression in the case $\mathcal{F} = \cc$.

With this simplified expression of the relative entropy of discord, we prove that this relative entropic measure is asymptotically continuous by repeatedly using the Fannes-Audenaert inequality~\cite{Audenaert2007} on the asymptotic continuity of the von Neumann entropy. 
Indeed, we show that for all states $\phi^{AB},\psi^{AB} \in \mathcal{D}(\mathcal{H}^{A}\otimes\mathcal{H}^{B})$, it holds that
\begin{align}~\label{theorem_discord_continuity}
    &|D^{\mathcal{F}}_{\textup{rel}}(\phi^{AB}) - D^{\mathcal{F}}_{\textup{rel}}(\psi^{AB})| \\
    &\leqq \|\phi^{AB} - \psi^{AB}\|_1 \log_2D + 2h_2\left(\frac{\|\phi^{AB} - \psi^{AB}\|_1}{2}\right),\nonumber
\end{align}
where $h_2$ is the binary entropy function defined as $h_2(x) \coloneqq -x\log_2x - (1-x)\log_2(1-x)$ for $x\in[0,1]$.  
Then, Eq.~\eqref{theorem_discord_continuity} implies asymptotic continuity of the relative entropy of discord. 
The rigorous proof is shown in the Supplemental Materials.
Considering Lemma~\ref{lemma_sufficient_condition} and Eq.~\eqref{theorem_discord_continuity} together, 
we conclude that the regularization of the relative entropy of discord is consistent for all the choices of free states. 

On the other hand, we also prove that another discord measure, \textit{measurement-based quantum discord}~\cite{Ollivier2001},
is subadditive and asymptotically continuous. 
Therefore, we can also employ its regularization as a consistent resource measure for discord. 
The detailed definition and proofs are shown in the Supplemental Materials.

\subsection{Non-Markovianity}~\label{subsec:non-Markovianity}
In this section, we prove that the regularized relative entropy of non-Markovianity is consistent. 
The quantum Markov chain is a quantum extension of the classical Markov chain, originally formulated in Ref.~\cite{Accardi1983}.  
In the classical case, a sequence of random variables $XYZ$ is said to be a Markov chain if $Z$ conditioned on $Y$ is independent of $X$, which is indeed equivalent to $I(X:Z|Y) = 0$, where $I(X:Z|Y)$ is the conditional mutual information~\cite{WYNER1978}. 
Analogously, a quantum state $\phi^{ABC} \in \mathcal{D}(\mathcal{H}^{A}\otimes \mathcal{H}^{B}\otimes \mathcal{H}^{C})$ is said to be a quantum Markov chain if $C$ conditioned on $B$ is independent of $A$~\cite{Accardi1983}, which is indeed equivalent to the condition $I(A:C|B)_{\phi} = 0$~\cite{Hayden2004} in terms of the quantum conditional mutual information~\cite{CERF1998}. 
Non-Markovianity serves as a resource in various tasks. 
For example, while a classical non-Markov chain enables the secret key agreement~\cite{Maurer1999}, 
quantum non-Markovianity is exploited as a resource in the quantum one-time pad~\cite{Sharma2020}, 
which is a protocol to ensure secure communication between two parties. 

The QRT of non-Markovianity has been established in Refs.~\cite{Wakakuwa2017,Wakakuwa2019} to analyze this quantum version of the Markov property in the QRT framework.
Hereafter, we consider the tripartite system $\mathcal{H}^{A} \otimes \mathcal{H}^{B} \otimes \mathcal{H}^{C}$ with dimension $D$ in total.
The set of quantum Markov chains $\mathcal{D}_{\mathrm{Markov}}$ is defined as $\mathcal{D}_{\mathrm{Markov}} \coloneqq \{\psi \in \mathcal{D}(\mathcal{H}^{A}\otimes \mathcal{H}^{B}\otimes \mathcal{H}^{C}): I(A:C|B)_{\psi} = 0\}$. 
Here, note that $\mathcal{D}_{\mathrm{Markov}}$ is not convex~\cite{Wakakuwa2017}; the QRT of non-Markovianity is a nonconvex QRT. 
Then, the relative entropy of non-Markovianity is defined as 
\begin{equation}
    \Delta(\phi^{ABC}) \coloneqq \min_{\psi \in \mathcal{D}_{\mathrm{Markov}}}  \relent(\phi^{ABC}\|\psi^{ABC}).
\end{equation}
We prove that the relative entropy of non-Markovianity is asymptotically continuous. 
We show that for all states $\phi, \psi \in \mathcal{D}(\mathcal{H}^A\otimes \mathcal{H}^B \otimes \mathcal{H}^C)$ satisfying $\|\phi-\psi\|_1/2 \leqq 1/3$,  
it holds that
\begin{align}~\label{theorem_markov_continuity}
    &|\Delta(\phi^{ABC}) - \Delta(\psi)^{ABC}|\\ 
    &\leqq 2(\|\phi^{ABC}-\psi^{ABC}\|_1\log_2 D + h_2(\|\phi^{ABC}-\psi^{ABC}\|_1)). \nonumber
\end{align}
Then, Eq.~\eqref{theorem_markov_continuity} implies that the relative entropy of non-Markovianity is asymptotically continuous. 
Our proof is based on a simple expression of the relative entropy of non-Markovianity, in a similar form to Eq.~\eqref{eq:relent_discord}, as a gap of two von Neumann entropies~\cite{Ibinson2008}. 
Then, we explicitly show asymptotic continuity of the relative entropy of non-Markovianity using this characterization. 
In the proof, 
we further decompose the characterization and repeatedly exploit Fannes-Audenaert inequality~\cite{Audenaert2007}. 
See the Supplemental Materials for more details of the proof.
Therefore, from Lemma~\ref{lemma_sufficient_condition} and Eq.~\eqref{theorem_markov_continuity}, it follows that the regularized relative entropy of non-Markovianity is consistent. 

\subsection{Non-Gaussianity}~\label{subsec:non-Gaussianity}
In this section, we investigate the consistency of the relative entropic measure in the QRT of non-Gaussianity~\cite{Takagi2018,Albarelli2018,Lami2018,Yamasaki2020}. 
Gaussian quantum information~\cite{Weedbrook2012}, based on Gaussian states and Gaussian operations, plays a central role in continuous-variable (CV) quantum information processing 
since Gaussian states and operations are experimentally implementable with high-precision control in quantum-optical setups~\cite{Walls2007}. 
In addition, despite the infinite-dimensionality, the theoretical analysis of Gaussian states is relatively tractable due to the fact that their characterization functions are in Gaussian forms. 
Indeed, various quantum information processing protocols, such as quantum teleportation~\cite{Vaidman1994,Braunstein1998,Ralph1998}, noisy quantum cloning~\cite{Cerf2000,Lindblad2000}, quantum illumination~\cite{Tan2008}, quantum reading~\cite{Pirandola2011}, and quantum key distribution~\cite{Grosshans2002,Grosshans2003,Weedbrook2004,Weedbrook2006,Matsuura2021}, can be implemented by Gaussian states and operations. 
However, it was shown that non-Gaussianity is an essential resource for several tasks including entanglement distillation~\cite{Eisert2002, Giedke2002, Fiuraifmmode2002, Zhang2010}, quantum error correction~\cite{Niset2009}, and universal quantum computation using CV systems~\cite{Lloyd1999,Bartlett2002,Ohliger2010,Menicucci2006,yamasaki2020polylogoverhead}.  
Non-Gaussianity also reveals the implementation cost of CV fault-tolerant quantum computation~\cite{Yamasaki2020, yamasaki2020polylogoverhead}.

Here, we briefly review the basic concepts for the QRT of non-Gaussianity. For further details, see Refs.~\cite{Braunstein2005,Weedbrook2012}.
Consider an $N$-mode bosonic system and define the real quadratic field operators $\hat{x} \coloneqq (\hat{q}_1,\hat{p}_1,\ldots,\hat{q}_N,\hat{p}_N)$. 
A quantum state $\phi$ is described by its Wigner characteristic function $\chi(\xi,\phi) = \tr[\phi \hat{D}(\xi)]$.
Here, $\xi = (\xi_1,\ldots,\xi_{2N}) \in \mathbb{R}^{2N}$ is a real-valued vector, and $\hat{D}(\xi) \coloneqq \exp(i \hat{x}^\mathrm{T}\Omega\xi)$ is the Weyl displacement operator with $\Omega \coloneqq i\bigoplus_{k=1}^N Y$, where $Y$ is the Pauli-$Y$ matrix.

A quantum state $\psi$ is a Gaussian state if its characteristic function $\chi(\xi,\psi)$ is expressed as a Gaussian function whose form is determined only by the mean and covariance matrix of $\hat{x}$ with respect to $\psi$~\cite{Braunstein2005,Weedbrook2012}. 
We define $\mathcal{G}$ as the set of Gaussian states. 
Since the sum of Gaussian functions is not necessarily Gaussian, the set $\mathcal{G}$ is not convex. 
The relative entropy of non-Gaussianity is defined as 
\begin{equation}
    \delta[\phi] \coloneqq \min_{\psi \in \mathcal{G}} \relent(\phi\|\psi).  
\end{equation}
It is known that the relative entropy of non-Gaussianity can be written as 
$\delta[\phi] = \entropy(\phi_{\textup{G}}) - \entropy(\phi)$, where $\phi_{\textup{G}}$ is the Gaussification of $\phi$, that is, the Gaussian state with the same mean and covariance matrix as $\phi$~\cite{Geroni2008}. 
Despite this simple characterization, we cannot directly apply the same strategy for asymptotic continuity as in the QRTs of discord and non-Markovianity because Gaussian states are defined on an infinite-dimensional state space. 

While continuity of the relative entropy of non-Gaussianity is claimed in Ref.~\cite{Geroni2008}, we discover a counterexample showing that the relative entropy of non-Gaussianity is not continuous even under a reasonable energy constraint. 
Note that this discovery may also affect some of the arguments in other existing literature citing Ref.~\cite{Geroni2008}.
Consider a single-mode system with Hamiltonian $H = \hbar\omega a^\dagger a$. 
Take two Fock-diagonal states $\rho = \epsilon \ket{E/\epsilon}\bra{E/\epsilon} + (1-\epsilon)\ket{0}\bra{0}$ and $\sigma = \ket{0}\bra{0}$ with a small positive number $\epsilon$ and fixed positive number $E$.
Then, these states satisfy the energy constraint $\tr(\rho H) \leqq E$ and $\tr(\sigma H) \leqq E$. 
However, while $\|\rho-\sigma\|_1 \approx \epsilon$, we have $|\delta[\rho] - \delta[\sigma]| \approx (E+1)\log_2(E+1) - E\log_2E$ even for infinitesimal $\epsilon$.
Therefore, the relative entropy of non-Gaussianity is not continuous. 
A more detailed discussion and proof are shown in the Supplemental Materials.

We also show that we can avoid this discontinuity of the relative entropy of non-Gaussianity, by considering the convex QRT of non-Gaussianity~\cite{Takagi2018,Albarelli2018}. 
In the convex QRT of non-Gaussianity, we take, as free operations, Gaussian operations that may be conditioned on outcomes of homodyne detections.
That is, we use the convex hull of the set of Gaussian operations as free operations. 
The set of free states is the convex hull of the set of Gaussian states $\mathrm{conv}(\mathcal{G})$.
Then, we consider the following modified version of relative entropy of non-Gaussianity, which we call the relative entropy of convex non-Gaussianity:
\begin{equation}
    \delta_{\textup{conv}}[\phi] \coloneqq \min_{\psi \in \mathrm{conv}(\mathcal{G})} \relent(\phi\|\psi). 
\end{equation}
Here, to avoid discontinuity due to the infinite dimensionality of the state space, we prove asymptotic continuity under an energy constraint. 
Hereafter, let $H$ be the Hamiltonian of a given system. 
We suppose that the smallest eigenvalue of $H$ is fixed to zero. 
Then, let $E\geqq 0$ represent an upper bound of the energy of the system. 
Let $\mathcal{D}_{H,E}(\mathcal{H}) \coloneqq \{\rho \in \mathcal{D}(\mathcal{H}): \tr(H\rho)\leqq E\}$ denote a subset of the set of density operators that contains all states satisfying the energy constraint, and we here take the set of states of interest as $\mathcal{S}(\mathcal{H}) = \mathcal{D}_{H,E}(\mathcal{H})$.
The relative entropy of convex non-Gaussianity $\delta_{\mathrm{conv}}$ is indeed asymptotically continuous in the following sense:
if the Hamiltonian $H$ satisfies the condition from Ref.~\cite{Shirokov2018}
\begin{equation}\label{eq:hamiltonian_hypothesis}
    \lim_{\lambda \to \infty} [\tr(\mathrm{e}^{-\lambda H})]^\lambda = 0, 
\end{equation}
for any sequences $\{\phi_n \in \mathcal{D}_{H_n,nE}(\mathcal{H})\}_n$ and $\{\psi_n \in \mathcal{D}_{H_n,nE}(\mathcal{H}) \}_n$ such that $\lim_{n\to\infty} \|\phi_n - \psi_n\|_1 \to 0$ 
with the Hamiltonian $H_n = H\otimes \I \otimes \cdots \otimes \I + \I \otimes H \otimes \I \otimes \cdots \otimes \I + \cdots + \I\otimes \cdots \otimes \I\otimes H$, 
then it holds that
$\lim_{n\to\infty} |\delta_{\textup{conv}}[\phi_n] - \delta_{\textup{conv}}[\psi_n]|/n = 0$. 

Here, we note that the assumption~\eqref{eq:hamiltonian_hypothesis}, on which asymptotic continuity is based, is satisfied by representative choices of Hamiltonians such as those of harmonic oscillators representing light.
Our proof of the asymptotic continuity is based on a modified version of Proposition 3 in Ref.~\cite{Shirokov2018} showing that
for the Hamiltonian $H$ satisfying Eq.~\eqref{eq:hamiltonian_hypothesis}
and for all states $\phi,\psi \in \mathcal{D}_{H,E}(\mathcal{H})$ such that $\|\phi-\psi\| \leqq \epsilon \leqq 1/2$, 
it holds that $|\delta_{\textup{conv}}[\phi] - \delta_{\textup{conv}}[\psi]| \to 0$ as $\epsilon \to 0$. 
Then, letting $r^{(H,E)}\left(\phi\to\psi\right)$ denote the asymptotic conversion rate under the energy constraint, 
we have the following inequality showing the consistency of $\delta^\infty_{\mathrm{conv}}$ 
\begin{equation}~\label{eq:AC_convex_Gaussian}
    \delta_{\textup{conv}}^{\infty}[\psi]r^{(H,E)}\left(\phi\to\psi\right) \leqq \delta_{\textup{conv}}^{\infty}[\phi].
\end{equation}
See the Supplemental Materials for the rigorous statement and proof.
Therefore, from Eq.~\eqref{eq:AC_convex_Gaussian}, it follows that the regularized relative entropy of convex non-Gaussianity is consistent in terms of the asymptotic state conversion under the energy constraint.

\section{Conclusion}
We investigated the relative entropy of resource for general resources, including discord, non-Markovianity, and non-Gaussianity 
to show the existence of asymptotically consistent resource measures in QRTs of these resources. 
In contrast with the analysis in our initial paper~\cite{Kuroiwa2020}, which showed the consistency of the regularized relative entropy of resource 
by assuming the finite dimensionality, convexity, and existence of full-rank free states, 
we here showed that the regularized relative entropy of resource is consistent in all finite-dimensional convex QRTs even \textit{without assuring the existence of full-rank free states}. 
Furthermore, we proved that the regularized relative entropy of resource serves as a consistent resource measure for some physically well-motivated \textit{nonconvex} quantum resources, namely, discord and non-Markovianity.
For an \textit{infinite-dimensional} resource, non-Gaussianity, we showed that the relative entropy of non-Gaussianity is indeed discontinuous even though the continuity is argued in the previous research~\cite{Geroni2008}. 
At the same time, by considering the convex hull of Gaussian states, we proved that the convex relative entropy of non-Gaussianity is asymptotically continuous under an energy constraint and that its regularization is consistent. 
Thus, we disclosed the consistency of the regularized relative entropy of resource for general resources beyond those covered in the previous work. 
These results pave the way for the quantification of resources consistent with the asymptotic state conversion, which fuels further investigation of resources not exactly comparable with each other. 
The results in this paper establish a theoretical foundation for quantitative studies of a wide variety of quantum resources that may not satisfy restrictive mathematical assumptions such as convexity and finite-dimensionality.

\begin{acknowledgments}
    This work was supported by JSPS Overseas Research Fellowships and JST, PRESTO Grant Number JPMJPR201A, Japan.
    K.\ K.\ was supported by Mike and Ophelia Lazaridis, and research grants by NSERC\@.
\end{acknowledgments}

\bibliography{CRM}
\clearpage
\widetext
\begin{center}
	\textbf{\large Supplemental Materials}
\end{center}
\setcounter{section}{0}
\setcounter{theorem}{0}
\setcounter{equation}{0}
\setcounter{figure}{0}
\setcounter{table}{0}
\setcounter{page}{1}
\renewcommand{\thesection}{S\Roman{section}}
\renewcommand{\theequation}{S\arabic{equation}}
\renewcommand{\thetheorem}{S\arabic{theorem}}
\renewcommand{\thefigure}{S\arabic{figure}}
\renewcommand{\bibnumfmt}[1]{[S#1]}
\renewcommand{\citenumfont}[1]{S#1}
In this supplemental document, we show detailed proofs of the statements asserted in the main paper. 
As mentioned in the main paper, we refer to \textit{asymptotically consistent measures} as just \textit{consistent measures}. 
In Sec.~\ref{SSec_Prelim}, we provide the notation throughout the paper and a fundamental lemma, which we use repeatedly. 
In Sec.~\ref{SSec_convex_finite}, we prove that the regularized relative entropy of resource is consistent in all finite-dimensional convex QRTs beyond the previous work~\cite{Kuroiwa2020} showing the consistency with the assumption of the existence of a full-rank free state. 
In Sec.~\ref{SSec_discord}, we investigate consistent resource measures for discord. 
There are three candidates for the set of free states, namely, classical-classical states, quantum-classical states, and classical-quantum states. 
In Sec.~\ref{SSubsec_relent_discord}, We prove the consistency of the regularized relative entropy of discord for any choice of free states. 
Considering the symmetry between quantum-classical states and classical-quantum states, 
we only show the consistency for classical-classical states (Sec.~\ref{SSubsec_CC}) and quantum-classical states (Sec.~\ref{SSubsec_QC}).
Furthermore, in Sec.~\ref{SSubsec_MBQD}, we investigate another discord measure, \textit{measurement-based quantum discord}, and we prove that this measure is also consistent. 
In Sec.~\ref{SSec_Markov}, we provide the proof for the consistency of the regularized relative entropy of non-Markovianity. 
In Sec.~\ref{SSec_NG}, we investigate the consistency of the regularized relative entropy of non-Gaussianity. 
We discover a counterexample revealing that the relative entropy of non-Gaussianity is not asymptotically continuous even under a reasonable energy constraint (See Sec.~\ref{SSubsec_nonconvex}). 
On the other hand, considering the convex hull of the set of Gaussian states, we show that the regularized relative entropy becomes consistent under an appropriate energy constraint (See Sec.~\ref{SSubsec_convex}). 

\section{Preliminaries}~\label{SSec_Prelim}
In this section, we provide the notation of this paper. 
Throughout this paper, we let $\mathcal{D}(\mathcal{H})$ denote the set of states on a quantum system $\mathcal{H}$.
Let $\entropy(\cdot)$ denote the quantum entropy defined as $\entropy(\phi) \coloneqq -\tr[\phi\log_2\phi]$ for a quantum state $\phi$.
Let $\relent(\cdot\|\cdot)$ denote the quantum relative entropy defined as follows. 
\begin{equation}
    \relent(\phi\|\psi) \coloneqq 
    \begin{cases}
    \tr(\phi\log_2\phi) - \tr(\phi\log_2\psi) & \mathrm{supp}(\phi) \subseteq \mathrm{supp}(\psi) \\
    +\infty & \mathrm{otherwise}
    \end{cases}
\end{equation}
where $\mathrm{supp}(\cdot)$ represents the support of an operator.
We define binary entropy function $h_2(\cdot)$ as
\begin{equation}
    h_2(x) \coloneqq -x\log_2x - (1-x)\log_2(1-x)
\end{equation}
on $x\in[0,1]$.

In this paper, we repeatedly use the following lemma known as Fannes-Audenaert inequality (Lemma~\ref{lemma_Fannes}).
\begin{lemma}[\cite{Audenaert2007}]~\label{lemma_Fannes}
    Let $\mathcal{H}$ be a finite-dimensional quantum system with dimension $D = \mathrm{dim}(\mathcal{H})$
    and let $\phi,\psi \in \mathcal{D}(\mathcal{H})$ be quantum states such that
    \begin{equation}
        \frac{\|\phi-\psi\|_1}{2} \leqq \epsilon
    \end{equation}
    for a fixed $\epsilon \in [0,1]$. Then, it holds that
    \begin{equation}
        |\entropy(\phi) - \entropy(\psi)| \leqq \epsilon\log_2D + h_2(\epsilon).
    \end{equation}
\end{lemma}

\section{Finite dimensional convex QRTs}~\label{SSec_convex_finite}
In this section, we argue that the regularized relative entropy of resource $R_\textup{R}$ is consistent in all finite-dimensional convex QRTs even without full-rank free states. 

Before going through the detailed analysis, we here summarize our proof. 
From the previous work~\cite{Kuroiwa2020}, the regularized relative entropy of resource is consistent in finite-dimensional convex QRTs with full-rank free states. 
Hereafter, suppose that the set of free states does not contain full-rank states; that is, 
the support of free states is strictly smaller than the whole system. 
Applying the same discussion as that with a full-rank free state, we show that the relative entropy of resource is asymptotically continuous on the support of free states. (See Theorem~\ref{finite_theorem1}.)
Then, by Sufficient Condition for Consistency (Lemma 3) in the main article, the regularized relative entropy of resource is consistent on the support of free states. 
Next, consider the case where $\phi$ is not in the support of free states. 
Proving that $\phi^{\otimes n}$ is also not in the support of free states (Lemma~\ref{finite_lemma2}), we observe that $R^{\infty}_R(\phi) = \infty$, and $R^{\infty}_R$ is consistent in terms of the refined definition of the consistency shown in the main text.
Finally, suppose that $\phi$ is in the support and $\psi$ is not; that is, $R^{\infty}_R(\phi) < \infty$ and $R^{\infty}_R(\psi) = \infty$. 
By using the same reasoning, we have that $\psi^{\otimes n}$ is not in the support of free states. 
Then, using the fact that free operations do not transfer a state from inside to outside of the support of free states (Lemma~\ref{finite_lemma1}), we show that we have $r(\phi \to \psi) = 0$ (Lemma~\ref{finite_lemma3} and Theorem~\ref{finite_theorem2}) in this case.  
Therefore, even without full-rank free states, the regularization of the relative entropy of resource is consistent in finite-dimensional convex QRTs under the refined definition of the consistent resource measures given in the main text.

Now, we provide detailed proof. 
We first show that the regularized relative entropy of resource is consistent on the support of free states. 
Here, we let $\mathcal{F}(\mathcal{H})\subseteq\mathcal{S}(\mathcal{H})$ denote the set of free states on the system $\mathcal{H}$.
We define the support of $\mathcal{F}(\mathcal{H})$, denoted by $\mathrm{supp}(\mathcal{F}(\mathcal{H}))$, as the smallest Hilbert space containing the supports of all free states in $\mathcal{F}(\mathcal{H})$; 
that is, 
\begin{equation}
    \mathrm{supp}(\mathcal{F}(\mathcal{H})) \coloneqq \mathrm{span}\left\{\ket{\psi} \in \mathcal{H}: \exists \phi \in \mathcal{F}(\mathcal{H})\,\,\mathrm{s.t.}\,\,\ket{\psi} \in \mathrm{supp}(\phi)\right\}. 
\end{equation}
\begin{lemma}~\label{finite_lemma2}
    Let $\mathcal{H}$ be a finite-dimensional quantum system. 
    Then, for any positive integer $n > 0$, $\mathrm{supp}(\mathcal{F}(\mathcal{H}^{\otimes n})) = \mathrm{supp}(\mathcal{F}(\mathcal{H}))^{\otimes n}$.
\end{lemma}

\begin{proof}
    The proof is by induction. 
    Let $D>0$ be the dimension of the system $\mathcal{H}$ and let $0< L \leqq D$ be the dimension of $\mathrm{supp}(\mathcal{F}(\mathcal{H}))$. 
    Let $\{\ket{1},\ldots,\ket{L}, \ket{L+1}, \ldots, \ket{D}\}$ be an orthonormal basis of $\mathcal{H}$ so that $\{\ket{1},\ldots,\ket{L}\}$ spans $\mathrm{supp}(\mathcal{F}(\mathcal{H}))$.
    For the proof by induction, suppose that $\mathrm{supp}(\mathcal{F}(\mathcal{H}^{\otimes n-1})) = \mathrm{supp}(\mathcal{F}(\mathcal{H}))^{\otimes n-1}$. 
    Here, take any state $\phi_n$ in $\mathcal{F}(\mathcal{H}^{\otimes n})$. 
    This state $\phi_n$ can be expressed as 
    \begin{equation}
        \phi_n \coloneqq \sum_{i_1,j_1,\ldots,i_n,j_n = 1,\ldots,D} \lambda_{i_1j_1,\cdots,i_nj_n}\ket{i_1\cdots i_n}\bra{j_1\cdots j_n}.
    \end{equation}
    Let $\mathcal{T}_{m}$ be the partial trace over $m$th system. 
    Since $\mathcal{T}_{m}$ is a free operation, $\mathcal{T}_{m}(\phi_n) \in \mathcal{F}(\mathcal{H}^{\otimes n-1})$ for all $1\leqq m\leqq n$.
    Since $\mathrm{supp}(\mathcal{F}(\mathcal{H}^{\otimes n-1})) = \mathrm{supp}(\mathcal{F}(\mathcal{H}))^{\otimes n-1}$, 
    $\mathcal{T}_{m}(\phi_n)$ can be expressed only by elements from $\{\ket{s_1,s_2,\ldots,s_{n-1}}:s_1,\ldots,s_{n-1} = 1,\ldots,L\}$.
    Now, take $m=1$. 
    Then, 
    \begin{equation}
        \mathcal{T}_{1}(\phi_n) \coloneqq \sum_{i_2,j_2,\ldots,i_n,j_n = 1,\ldots,D}\left(\sum_{i_1} \lambda_{i_1i_1i_2j_2\cdots i_nj_n}\right) \ket{i_2\cdots i_n}\bra{j_2\cdots j_n}.
    \end{equation}
    Since $\mathcal{T}_{1}(\phi_n)$ must not consist of any state in $\{\ket{L+1},\ldots,\ket{D}\}$, 
    if either one of $i_2, \ldots, i_n$ is in $\{L+1,\ldots,D\}$,
    it holds that 
    \begin{equation}
        \sum_{i_1} \lambda_{i_1i_1i_2i_2\cdots i_ni_n} = 0.
    \end{equation}
    Since $\lambda_{i_1i_1i_2i_2\cdots i_ni_n} \geqq 0$ for all $i_1$, in this case, $\lambda_{i_1i_1i_2i_2\cdots i_ni_n} = 0$ for all $i_1$. 
    Therefore, $\phi_n$ does not contain a state from $\{L+1,\ldots,D\}$ in all systems except the first one. 
    The same argument also applies to arbitrary choice of $m \in \{1,2,\ldots,n\}$. 
    Therefore, we have that $\phi_n$ does not contain a state from $\{L+1,\ldots,D\}$ in all systems.
    Hence, it follows that $\phi_n$ does not contain a state from $\{L+1,\ldots,D\}$, which implies that $\{\ket{s_1,s_2,\ldots,s_{n}}:s_1,\ldots,s_{n} = 1,\ldots,L\}$ spans $\mathcal{F}(\mathcal{H}^{\otimes n})$; 
    that is, $\mathrm{supp}(\mathcal{F}(\mathcal{H}^{\otimes n})) = \mathrm{supp}(\mathcal{F}(\mathcal{H}))^{\otimes n}$.
\end{proof}

\begin{theorem}\label{finite_theorem1}
    Let $\phi$ and $\psi$ be states such that $R_\textup{R}^{\infty}(\phi)<\infty$ and  $0 < R_\textup{R}^{\infty}(\psi) < \infty$. 
    Then, it holds that $r(\phi\to\psi) < \infty$. 
    Moreover, $R_\textup{R}^{\infty}$ is consistent for these states; that is, $R_\textup{R}^{\infty}(\psi)r(\phi\to\psi) \leqq R_\textup{R}^{\infty}(\phi)$. 
\end{theorem}

\begin{proof}
    By Lemma~\ref{finite_lemma2}, we have $R_\textup{R}(\phi), R_\textup{R}(\psi) < \infty$, and thus $\phi$ and $\psi$ are in the support of free states. 
    Here, we note that the relative entropy of resource is asymptotically continuous on the support of free states using the same argument of the proof in Refs.~\cite{Synak2006,Winter2016}. 

    We first prove that $r(\phi\to\psi) < \infty$. The proof is by contradiction. 
    Assume that $r(\phi\to\psi) = \infty$. 
    Then, we can take a positive number 
    \begin{equation}~\label{eq:def_C}
    C > \frac{R_\textup{R}^{\infty}(\phi)}{R_\textup{R}^{\infty}(\psi)}, 
    \end{equation}
    and for an arbitrarily small $\epsilon > 0$, there exists sufficiently large $n$ and a free operation $\mathcal{N}_n$ such that
    \begin{equation}
        \left\|\mathcal{N}_n(\phi^{\otimes n}) -  \psi^{\otimes \lceil Cn \rceil}\right\|_1 \leqq \epsilon. 
    \end{equation} 
    Using asymptotic continuity of $R_\textup{R}$ on the support of free states, we have
    \begin{equation}
        R_\textup{R}(\phi^{\otimes n}) \geqq R_\textup{R}(\psi^{\otimes \lceil Cn \rceil}) - \lceil Cn \rceil\epsilon. 
    \end{equation}
    Then, dividing this inequality by $\lceil Cn \rceil$, we have
    \begin{equation}
        \frac{R_\textup{R}(\phi^{\otimes n})}{Cn} \geqq \frac{R_\textup{R}(\phi^{\otimes n})}{\lceil Cn \rceil} \geqq \frac{R_\textup{R}(\psi^{\otimes \lceil Cn \rceil})}{\lceil Cn \rceil} - \epsilon. 
    \end{equation}
    Taking the limit $n\to\infty$ and $\epsilon \to 0$, we have
    \begin{equation}
        R_\textup{R}^{\infty}(\phi) \geqq CR_\textup{R}^{\infty}(\psi), 
    \end{equation}
    which contradicts the assumption~\eqref{eq:def_C}. 

    Next, we prove the consistency. 
    Since $r(\phi\to\psi) < \infty$, for an arbitrarily small $\delta > 0$, $r \coloneqq r(\phi\to\psi) - \delta$ is an achievable rate. 
    Then, for an arbitrarily small $\epsilon > 0$, there exists sufficiently large $n$ and a free operation $\mathcal{N}_n$ such that
    \begin{equation}
        \left\|\mathcal{N}_n(\phi^{\otimes n}) -  \psi^{\otimes \lceil rn \rceil}\right\|_1 \leqq \epsilon. 
    \end{equation} 
    Using asymptotic continuity of $R_\textup{R}$ on the support of free states, we have
    \begin{equation}
        R_\textup{R}(\phi^{\otimes n}) \geqq R_\textup{R}(\psi^{\otimes \lceil rn \rceil}) - \lceil Cn \rceil\epsilon. 
    \end{equation}
    Then, dividing this inequality by $\lceil rn \rceil$, we have
    \begin{equation}
        \frac{R_\textup{R}(\phi^{\otimes n})}{rn} \geqq \frac{R_\textup{R}(\phi^{\otimes n})}{\lceil rn \rceil} \geqq \frac{R_\textup{R}(\psi^{\otimes \lceil rn \rceil})}{\lceil rn \rceil} - \epsilon. 
    \end{equation}
    Taking the limit $n\to\infty$ and $\epsilon \to 0$, we have
    \begin{equation}
        R_\textup{R}^{\infty}(\phi) \geqq rR_\textup{R}^{\infty}(\psi). 
    \end{equation}
    Since we can take $\delta$ arbitrarily small, $R_\textup{R}$ is consistent. 
\end{proof}

Then, we show that the regularized relative entropy of resource quantifies resources consistently with the asymptotic conversion rate even for resources out of the support of free states. 

\begin{lemma}~\label{finite_lemma1}
    Let $\mathcal{H}$ and $\mathcal{H}'$ be quantum systems. 
    Suppose $\phi \in \mathcal{D}(\mathcal{H})$ is a state such that $\mathrm{supp}(\phi) \subseteq \mathrm{supp}(\mathcal{F}(\mathcal{H}))$. 
    Then, for an arbitrary free operation $\mathcal{N}$ from a system $\mathcal{H}$ to a system $\mathcal{H}'$, it holds that $\mathrm{supp}(\mathcal{N}(\phi)) \subseteq \mathrm{supp}(\mathcal{F}(\mathcal{H}'))$. 
\end{lemma}
\begin{proof}
    Since $\mathrm{supp}(\phi) \subseteq \mathrm{supp}(\mathcal{F}(\mathcal{H}))$, 
    there exists a free state $\psi$ such that
    \begin{equation}
        \relent(\phi\|\psi) < \infty.  
    \end{equation}
    Using the monotonicity of the relative entropy, it holds that
    \begin{equation}
        \relent(\mathcal{N}(\phi)\|\mathcal{N}(\psi)) \leqq  \relent(\phi\|\psi) < \infty.
    \end{equation}
    Since $\mathcal{N}(\psi)$ is also a free state, we have $\mathrm{supp}(\mathcal{N}(\phi)) \subseteq \mathrm{supp}(\mathcal{N}(\psi)) \subseteq \mathrm{supp}(\mathcal{F}(\mathcal{H}'))$.
\end{proof}

\begin{lemma}~\label{finite_lemma3}
    Let $\mathcal{H}$ and $\mathcal{H}'$ be quantum systems. 
    Let $\phi \in \mathcal{D}(\mathcal{H})$ such that $\mathrm{supp}(\phi) \subseteq \mathrm{supp}(\mathcal{F}(\mathcal{H}))$, 
    and let $\psi \in \mathcal{D}(\mathcal{H}')$ such that $\mathrm{supp}(\psi) \not\subseteq \mathrm{supp}(\mathcal{F}(\mathcal{H}'))$. 
    Then, the asymptotic conversion rate $r(\phi \to \psi)$ is zero. 
\end{lemma}

\begin{proof}
    The proof is by contradiction. 
    Suppose that $r(\phi \to \psi) > 0$. 
    Then, for some $r>0$, it holds that for an arbitrarily small $\epsilon$, there exists sufficiently large $n$ and a free operation $\mathcal{N}_n$ such that
    \begin{equation}
        \left\|\mathcal{N}_n(\phi^{\otimes n}) - \psi^{\otimes \lceil rn \rceil}\right\|_1 < \epsilon. 
    \end{equation}
    Let $\mathcal{T}$ denote the partial trace over all systems other than the first one. 
    Then, using the monotonicity of the trace norm, we have
    \begin{equation}
        \left\|\mathcal{T}\circ\mathcal{N}_n(\phi^{\otimes n}) - \psi\right\|_1 \leqq \left\|\mathcal{N}_n(\phi^{\otimes n}) - \psi^{\otimes \lceil rn \rceil}\right\|_1 
    \end{equation}
    for any $n>0$ and any free operation $\mathcal{N}_n$. 
    From Lemma~\ref{finite_lemma1}, $\mathcal{T}\circ\mathcal{N}_n(\phi^{\otimes n})$ is also in the support of $\mathcal{F}(\mathcal{H}')$.
    On the other hand, $\mathrm{supp}(\psi) \not\subseteq \mathrm{supp}(\mathcal{F}(\mathcal{H}'))$. 
    Therefore, $\left\|\mathcal{T}\circ\mathcal{N}_n(\phi^{\otimes n}) - \psi\right\|_1$ can be lower-bounded by 
    \begin{equation}
        D_\psi \coloneqq \min_{\sigma} \|\sigma - \psi\|_1 > 0, 
    \end{equation} 
    where the minimization is taken over all states $\sigma$ in the support of $\mathcal{F}(\mathcal{H}')$.
    If we take $\epsilon$ smaller than $D_\psi$, it always holds that 
    \begin{equation}
        \epsilon \leqq \left\|\mathcal{N}_n(\phi^{\otimes n}) - \psi^{\otimes \lceil rn \rceil}\right\|_1, 
    \end{equation}
    which contradicts to the assumption of $r(\phi\to\psi) > 0$. 
\end{proof}

\begin{theorem}\label{finite_theorem2}
    Let $\mathcal{H}$ and $\mathcal{H}'$ be quantum systems. 
    Let $\phi \in \mathcal{D}(\mathcal{H})$ such that $R_\textup{R}^{\infty}(\phi) < \infty$, 
    and let $\psi \in \mathcal{D}(\mathcal{H}')$ such that $R_\textup{R}^{\infty}(\psi) = \infty$. 
    Then, it holds that $r(\phi\to\psi) = 0$. 
\end{theorem}

\begin{proof}
    From Lemma~\ref{finite_lemma2}, since $R_\textup{R}^{\infty}(\phi) < \infty$, we have $\mathrm{supp}(\phi) \subseteq \mathrm{supp}(\mathcal{F}(\mathcal{H}))$. 
    Since $R_\textup{R}^{\infty}(\psi) = \infty$, it holds that $\mathrm{supp}(\psi) \not\subseteq \mathrm{supp}(\mathcal{F}(\mathcal{H}'))$.
    Then, by Lemma~\ref{finite_lemma3}, we have $r(\phi\to\psi) = 0$. 
\end{proof}

\section{Discord}~\label{SSec_discord}
In this section, we investigate consistent resource measures in the resource theory of discord. 
We basically adopt the notation in Ref.~\cite{Bera2017}.
In Sec~\ref{SSubsec_relent_discord}, we study consistency of the regularized relative entropy of discord. 
We first consider classical-classical states as free states in Sec.~\ref{SSubsec_CC}, and we show that the regularized relative entropy of discord is consistent for this choice of free states. 
Then, in Sec.~\ref{SSubsec_QC}, we consider quantum-classical states as free states, and we also show that the regularized relative entropy of discord is consistent in this case.
Note that due to the symmetry, the proof for classical-quantum states is also applicable to quantum-classical states with the quantum and classical registers switched. 
In Sec.~\ref{SSubsec_MBQD}, we investigate another discord measure, \textit{measurement-based quantum discord}, and we prove that this measure is also consistent. 
\subsection{Relative Entropy of Discord}~\label{SSubsec_relent_discord}
In this section, we investigate the relative entropy of discord for two choices of free states, namely, classical-classical states and quantum-classical states. 
We prove that the regularized relative entropy of discord is consistent for both of these two cases. 
\subsubsection{Classical-Classical States}~\label{SSubsec_CC}
Here, we formally define the set of classical-classical states. 
\begin{definition}
    Let $\mathcal{H}^A$ and $\mathcal{H}^B$ be finite-dimensional quantum systems.
    Then, we define a set $\cc$ of classical states on $\mathcal{H}^A\otimes\mathcal{H}^B$ as
    \begin{equation}
        \cc \coloneqq \left\{\sum_{k}p_{k}\ket{a_k}\bra{a_k}^{A}\otimes\ket{b_k}\bra{b_k}^{B}\right\},
    \end{equation}
    where $\ket{a_k}^A$ and $\ket{b_k}^B$ are orthonormal bases of systems $\mathcal{H}^A$ and $\mathcal{H}^B$ respectively, 
    and $p_{k}$ is a probability distribution; that is, $\cc$ is the set of the mixture of local classical states.  
\end{definition}
\begin{remark}
    $\cc$ is not convex.
\end{remark}
The relative entropy of discord for $\cc$ is defined as follows. 
\begin{definition}
    Let $\mathcal{H}^A$ and $\mathcal{H}^B$ be finite-dimensional quantum systems.
    For a quantum state $\phi^{AB} \in \mathcal{D}(\mathcal{H}^A\otimes\mathcal{H}^B)$, 
    define the relative entropy of discord $D^{\cc}_{\textup{rel}}$ as
    \begin{equation}
        D^{\cc}_{\textup{rel}}(\phi^{AB}) \coloneqq \min_{\psi^{AB} \in \cc} \relent(\phi^{AB}\|\psi^{AB}).
    \end{equation}
\end{definition}

Then, we derive a simple characterization of the relative entropy of discord for $\cc$. 
\begin{proposition}\label{Slemma_relent_cc}
    Let $\mathcal{H}^A$ and $\mathcal{H}^B$ be finite-dimensional quantum systems.
    For a quantum state $\phi^{AB}\in \mathcal{D}(\mathcal{H}^A\otimes\mathcal{H}^B)$, 
    it holds that
    \begin{equation}\label{cc_statement}
        D^{\cc}_{\textup{rel}}(\phi^{AB}) = \min_{\ket{a_k}^A,\ket{b_k}^B} \entropy\left(\sum_k(\bra{a_k}\otimes\bra{b_k})\phi^{AB}(\ket{a_k}\otimes\ket{b_k})\ket{a_k}\bra{a_k}^A\otimes\ket{b_k}\bra{b_k}^B\right) - \entropy(\phi^{AB}),
    \end{equation}
    where the minimization is taken over all orthonormal bases of $\mathcal{H}^A$ and $\mathcal{H}^B$. 
\end{proposition}

\begin{proof}
    Take a classical-classical state
    \begin{equation}
        \psi^{AB} \coloneqq \sum_k p_k \ket{a_k}\bra{a_k}^A\otimes\ket{b_k}\bra{b_k}^B. 
    \end{equation}
    such that $\mathrm{supp}(\phi^{AB}) \subseteq \mathrm{supp}(\sigma^{AB})$.
    Then, it holds that
    \begin{align}
        \relent(\phi^{AB}\|\psi^{AB})
        &= -\entropy(\phi^{AB}) - \tr[\phi^{AB}\log_2\psi^{AB}]\\
        &= -\entropy(\phi^{AB}) - \tr\left[\phi^{AB}\log_2\left(\sum_k p_k \ket{a_k}\bra{a_k}^A\otimes\ket{b_k}\bra{b_k}^B\right)\right]\\
        &= -\entropy(\phi^{AB}) - \sum_k\tr\left[\phi^{AB}\left(\log_2p_k\ket{a_k}\bra{a_k}^A\otimes\ket{b_k}\bra{b_k}^B\right)\right] \\
        &= -\entropy(\phi^{AB}) - \sum_k\log_2p_k(\bra{a_k}\otimes\bra{b_k})\phi^{AB}(\ket{a_k}\otimes\ket{b_k}).
    \end{align}
    Now, define $q_k \coloneqq (\bra{a_k}\otimes\bra{b_k})\phi^{AB}(\ket{a_k}\otimes\ket{b_k})$.
    Then, we have
    \begin{align}
        \relent(\phi^{AB}\|\psi^{AB})
        &= -\entropy(\phi^{AB}) - \sum_k q_k\log_2p_k\\
        &= -\entropy(\phi^{AB}) -\sum_k q_k\log_2q_k + \left(\sum_k q_k\log_2q_k - \sum_k q_k\log_2p_k\right) \\
        \label{cc_ineq_1}
        &\geqq -\entropy(\phi^{AB}) -\sum_k q_k\log_2q_k \\
        &= -\entropy(\phi^{AB}) + \entropy\left(\sum_k(\bra{a_k}\otimes\bra{b_k})\phi^{AB}(\ket{a_k}\otimes\ket{b_k})\ket{a_k}\bra{a_k}^A\otimes\ket{b_k}\bra{b_k}^B\right).
    \end{align}
    The inequality~\eqref{cc_ineq_1} follows from the positivity of the relative entropy, and it is saturated by taking $p_k = q_k$ for all $k$. 
    Therefore, the minimum of $\mathrm{D}(\phi^{AB}\|\psi^{AB})$ can be achieved by choosing optimal $\ket{a_k}^A$ and $\ket{b_k}^B$, 
    which implies \eqref{cc_statement}.
\end{proof}

Then, applying Proposition~\ref{Slemma_relent_cc}, we prove the asymptotic continuity of $D^{\cc}_{\textup{rel}}$.
\begin{theorem}
    Let $\mathcal{H}^A$ and $\mathcal{H}^B$ be finite-dimensional quantum systems with dimension $D = \mathrm{dim}(\mathcal{H}^A\otimes\mathcal{H}^B)$.
    The relative entropy of discord $D^{\cc}_{\textup{rel}}$ is asymptotically continuous; \textit{i.e.}, 
    it holds that
    \begin{equation}
        |D^{\cc}_{\textup{rel}}(\phi^{AB}) - D^{\cc}_{\textup{rel}}(\psi^{AB})| \leqq \|\phi^{AB} - \psi^{AB}\|_1 \log_2D + 2h_2\left(\frac{\|\phi^{AB} - \psi^{AB}\|_1}{2}\right), 
    \end{equation}
    where $h_2$ is the binary entropy function. 
\end{theorem}
\begin{proof}
By Proposition~\ref{Slemma_relent_cc}, it holds that
\begin{align*}
    |D^{\cc}_{\textup{rel}}(\phi^{AB}) - D^{\cc}_{\textup{rel}}(\psi^{AB})|
        &\leqq |\entropy(\phi^{AB})-\entropy(\psi^{AB})|\\
        &+ \Bigg|\entropy\left(\sum_k(\bra{a_k}\otimes\bra{b_k})\phi^{AB}(\ket{a_k}\otimes\ket{b_k})\ket{a_k}\bra{a_k}^A\otimes\ket{b_k}\bra{b_k}^B\right)\\
        &\quad- \entropy\left(\sum_k(\bra{a_k}\otimes\bra{b_k})\psi^{AB}(\ket{a_k}\otimes\ket{b_k})\ket{a_k}\bra{a_k}^A\otimes\ket{b_k}\bra{b_k}^B\right)\Bigg|,
\end{align*}
where $\ket{a_k}^A$ and $\ket{b_k}^B$ is optimal orthonormal bases for $\psi$. 
Here, we assume that 
\begin{equation*}
    \begin{aligned}
        \min_{\ket{a_k},\ket{b_k}} &\entropy\left(\sum_k(\bra{a_k}\otimes\bra{b_k})\phi^{AB}(\ket{a_k}\otimes\ket{b_k})\ket{a_k}\bra{a_k}^A\otimes\ket{b_k}\bra{b_k}^B\right) \\
        &\geqq \min_{\ket{a_k},\ket{b_k}} \entropy\left(\sum_k(\bra{a_k}\otimes\bra{b_k})\psi^{AB}(\ket{a_k}\otimes\ket{b_k})\ket{a_k}\bra{a_k}^A\otimes\ket{b_k}\bra{b_k}^B\right).
    \end{aligned}
\end{equation*}
Otherwise, we can switch the roles of $\phi^{AB}$ and $\psi^{AB}$ to apply the same argument. 

Now, it holds that
\begin{align*}
    &\left\|\sum_k(\bra{a_k}\otimes\bra{b_k})\phi^{AB}(\ket{a_k}\otimes\ket{b_k})\ket{a_k}\bra{a_k}^A\otimes\ket{b_k}\bra{b_k}^B - \sum_k(\bra{a_k}\otimes\bra{b_k})\psi^{AB}(\ket{a_k}\otimes\ket{b_k})\ket{a_k}\bra{a_k}^A\otimes\ket{b_k}\bra{b_k}^B\right\|_1\\
    &=\sum_k\left\|\left((\bra{a_k}\otimes\bra{b_k})\left(\phi^{AB} - \psi^{AB}\right)(\ket{a_k}^A\otimes\ket{b_k}^B)\right)\ket{a_k}\bra{a_k}^A\otimes\ket{b_k}\bra{b_k}^B\right\|_1\\
    &=\left\|\sum_k\left((\bra{a_k}\otimes\bra{b_k})\left(\phi^{AB} - \psi^{AB}\right)(\ket{a_k}^A\otimes\ket{b_k}^B)\right)\ket{a_k}\bra{a_k}^A\otimes\ket{b_k}\bra{b_k}^B\right\|_1\\
    &\leqq \|\phi^{AB}-\psi^{AB}\|_1. 
\end{align*}
Therefore, using Lemma~\ref{lemma_Fannes} (Fannes-Audenaert inequality), we have
    \begin{equation}
        |D^{\cc}_{\textup{rel}}(\phi^{AB}) - D^{\cc}_{\textup{rel}}(\psi^{AB})| \leqq \|\phi^{AB} - \psi^{AB}\|_1 \log_2D + 2h_2\left(\frac{\|\phi^{AB} - \psi^{AB}\|_1}{2}\right). 
    \end{equation}
\end{proof}

\subsubsection{Quantum-Classical States}~\label{SSubsec_QC}
First, we formally define the set of classical-classical states. 
\begin{definition}
    Let $\mathcal{H}^A$ and $\mathcal{H}^B$ be finite-dimensional quantum systems.
    Then, we define a set $\qc$ of quantum-classical states on $\mathcal{H}^A\otimes\mathcal{H}^B$ as
    \begin{equation}
        \qc \coloneqq \left\{\sum_{l}p_{l}\rho_l^A\otimes\ket{l}\bra{l}^B\right\},
    \end{equation}
    where and $\ket{l}^B$ are orthonormal bases of systems $\mathcal{H}^A$ and $\mathcal{H}^B$ respectively, 
    and $p_{l}$ is a probability distribution.
\end{definition}
\begin{remark}
    $\qc$ is not convex.
\end{remark}
Then, we define the relative entropy of discord for $\qc$. 
\begin{definition}
    Let $\mathcal{H}^A$ and $\mathcal{H}^B$ be finite-dimensional quantum systems.
    For a quantum state $\phi^{AB} \in \mathcal{D}(\mathcal{H}^A\otimes\mathcal{H}^B)$, 
    define the relative entropy of discord $D^{\qc}_{\textup{rel}}$ as
    \begin{equation}
        D^{\qc}_{\textup{rel}}(\phi^{AB}) \coloneqq \min_{\psi^{AB} \in \qc} \relent(\phi^{AB}\|\psi^{AB}).
    \end{equation}
\end{definition}

Here, we give a simple characterization of the relative entropy of discord for $\qc$. 
\begin{proposition}\label{Slemma_relent_qc}
    Let $\mathcal{H}^A$ and $\mathcal{H}^B$ be finite-dimensional quantum systems.
    For a quantum state $\phi^{AB} \in \mathcal{D}(\mathcal{H}^A\otimes\mathcal{H}^B)$, 
    it holds that
    \begin{equation}\label{qc_statement}
        D^{\qc}_{\textup{rel}}(\phi^{AB}) = \min_{\ket{b_k}^B} \entropy\left(\sum_k(\I^A\otimes\bra{b_k}^B)\phi^{AB}(\I^A\otimes\ket{b_k}^B)\otimes\ket{b_k}\bra{b_k}^B\right) - \entropy(\phi^{AB}),
    \end{equation}
    where the minimization is taken over all orthonormal bases of $\mathcal{H}^B$. 
\end{proposition}

\begin{proof}
    Take a quantum-classical state
    \begin{equation}
        \psi^{AB} \coloneqq \sum_k p_k \rho^A_k\otimes\ket{b_k}\bra{b_k}^B. 
    \end{equation}
    such that $\mathrm{supp}(\phi^{AB}) \subseteq \mathrm{supp}(\psi^{AB})$.
    Then, it holds that
    \begin{align}
        &\relent(\phi^{AB}\|\psi^{AB})\\
        &= -\entropy(\phi^{AB}) - \tr[\phi^{AB}\log_2\psi^{AB}]\\
        &= -\entropy(\phi^{AB}) - \tr\left[\phi^{AB}\log_2\left(\sum_k p_k \rho^A_k\otimes\ket{b_k}\bra{b_k}^B\right)\right]\\
        &= -\entropy(\phi^{AB}) - \sum_k\tr\left[\phi^{AB}\left(\log_2\rho^A_k\otimes\ket{b_k}\bra{b_k}^B\right)\right] - \sum_k\tr\left[\phi^{AB}\left(\I_{\mathrm{supp}(\rho^A_k)}\otimes\log_2p_k\ket{b_k}\bra{b_k}^B\right)\right]\\
        &= -\entropy(\phi^{AB}) - \sum_k\tr\left[(\I^A\otimes\bra{b_k}^B)\phi^{AB}(\I^A\otimes\ket{b_k}^B)\log_2\rho^A_k\right] - \sum_k \tr[(\I^A\otimes\bra{b_k}^B)\phi^{AB}(\I^A\otimes\ket{b_k}^B)]\log_2p_k.
    \end{align}
    Now, write
    \begin{align}
        &q_k \coloneqq \tr\left[(\I^A\otimes\bra{b_k}^B)\phi^{AB}(\I^A\otimes\ket{b_k}^B)\right],\\
        &\phi^A_k \coloneqq \frac{(\I^A\otimes\bra{b_k}^B)\phi^{AB}(\I^A\otimes\ket{b_k}^B)}{q_k}.
    \end{align}
    Then, we have
    \begin{align}
        \relent(\phi^{AB}\|\psi^{AB})
        &= -\entropy(\phi^{AB}) - \sum_kq_k\tr\left[\phi^A_k\log_2\rho^A_k\right] -\sum_k q_k\log_2p_k\\
        &= -\entropy(\phi^{AB}) -\sum_k q_k\log_2q_k + \left(\sum_k q_k\log_2q_k -\sum_k q_k\log_2p_k\right) + \sum_k q_k \entropy(\phi^A_k) + \sum_kq_k\relent(\phi^A_k\|\rho^A_k)\\
        \label{qc_ineq_1}
        &\geqq -\entropy(\phi^{AB}) -\sum_k q_k\log_2q_k + \sum_k q_k \entropy(\phi^A_k)\\
        &= -\entropy(\phi^{AB}) + \entropy\left(\sum_kq_k\phi^A_k\otimes\ket{b_k}\bra{b_k}^B\right)\\
        \label{qc_ineq_2}
        &= \entropy\left(\sum_k(\I^A\otimes\bra{b_k}^B)\rho(\I^A\otimes\ket{b_k}^B)\otimes\ket{b_k}\bra{b_k}^B\right) - \entropy(\phi^{AB}).
    \end{align}
    The inequality~\eqref{qc_ineq_1} follows from the positivity of the relative entropy, and it is saturated by taking $p_k = q_k$ and $\rho^A_k = \phi^A_k$ for all $k$. 
    Therefore, the minimum of $\relent(\phi^{AB}\|\psi^{AB})$ can be achieved by choosing optimal $\ket{b_k}^B$ in \eqref{qc_ineq_2}, 
    which implies \eqref{qc_statement}.
\end{proof}

Using Proposition~\ref{Slemma_relent_qc}, we prove the asymptotic continuity of $D^{\qc}_{\textup{rel}}$.
\begin{theorem}
    Let $\mathcal{H}^A$ and $\mathcal{H}^B$ be finite-dimensional quantum systems with dimension $D = \mathrm{dim}(\mathcal{H}^A\otimes\mathcal{H}^B)$.
    The relative entropy of discord $D^{\qc}_{\textup{rel}}$ is asymptotically continuous; \textit{i.e.}, 
    it holds that
    \begin{equation}
        |D^{\qc}_{\textup{rel}}(\phi^{AB}) - D^{\qc}_{\textup{rel}}(\psi^{AB})| \leqq \|\phi^{AB} - \psi^{AB}\|_1 \log_2D + 2h_2\left(\frac{\|\phi^{AB} - \psi^{AB}\|_1}{2}\right), 
    \end{equation}
    where $h_2$ is the binary entropy function. 
\end{theorem}
\begin{proof}
By Proposition~\ref{Slemma_relent_qc}, it holds that
\begin{align*}
    |D^{\qc}_{\textup{rel}}(\phi^{AB}) - D^{\qc}_{\textup{rel}}(\psi^{AB})|
        &\leqq |\entropy(\phi^{AB})-\entropy(\psi^{AB})| \nonumber \\
        &+ \Bigg|\entropy\left(\sum_k(\I^A\otimes\bra{b_k}^B)\phi^{AB}(\I^A\otimes\ket{b_k}^B)\otimes\ket{b_k}\bra{b_k}^B\right)\\
        &\quad- \entropy\left(\sum_k(\I^A\otimes\bra{b_k}^B)\psi^{AB}(\I^A\otimes\ket{b_k}^B)\otimes\ket{b_k}\bra{b_k}^B\right)\Bigg|,
\end{align*}
where $\ket{b_k}^B$ is an optimal orthonormal basis for $\psi^{AB}$. 
Here, we assume that 
\begin{equation*}
    \min_{\ket{b_k}^B} \entropy\left(\sum_k(\I^A\otimes\bra{b_k}^B)\phi^{AB}(\I^A\otimes\ket{b_k}^B)\otimes\ket{b_k}\bra{b_k}^B\right) \geqq \min_{\ket{b_k}} \entropy\left(\sum_k(\I^A\otimes\bra{b_k}^B)\psi^{AB}(\I^A\otimes\ket{b_k}^B)\otimes\ket{b_k}\bra{b_k}^B\right).    
\end{equation*}
Otherwise, we can switch the roles of $\phi^{AB}$ and $\psi^{AB}$ to apply the same argument. 
Now, it holds that
\begin{align}
    &\left\|\sum_k(\I^A\otimes\bra{b_k}^B)\phi^{AB}(\I^A\otimes\ket{b_k}^B)\otimes\ket{b_k}\bra{b_k}^B - \sum_k(\I^A\otimes\bra{b_k}^B)\psi^{AB}(\I^A\otimes\ket{b_k}^B)\otimes\ket{b_k}\bra{b_k}^B\right\|_1\\
    &=\sum_k\left\|(\I^A\otimes\bra{b_k}^B)\left(\phi^{AB} - \psi^{AB}\right)(\I^A\otimes\ket{b_k}^B)\right\|_1\\
    &= \left\|\sum_k(\I^A\otimes\bra{b_k}^B)\left(\phi^{AB} - \psi^{AB}\right)(\I^A\otimes\ket{b_k}^B)\right\|_1\\
    &\leqq \|\phi^{AB}-\psi^{AB}\|_1. 
\end{align}
Therefore, using Lemma~\ref{lemma_Fannes} (Fannes-Audenaert inequality), we have
    \begin{equation}
        |D^{\qc}_{\textup{rel}}(\phi^{AB}) - D^{\qc}_{\textup{rel}}(\psi^{AB})| \leqq \|\phi^{AB} - \psi^{AB}\|_1 \log_2D + 2h_2\left(\frac{\|\phi^{AB} - \psi^{AB}\|_1}{2}\right). 
    \end{equation}
\end{proof}

\subsection{Measurement-Based Quantum Discord}~\label{SSubsec_MBQD}
In this section, we show that another measure for discord, \textit{measurement-based quantum discord}, also satisfies subadditivity and asymptotic continuity; that is, its regularization is consistent.
\begin{definition}[Measurement-Based Quantum Discord~\cite{Ollivier2001}]
    Let $\mathcal{H}^A$ and $\mathcal{H}^B$ be finite-dimensional quantum systems.
    For a quantum state $\phi^{AB}$, the measurement-based quantum discord for the measurement on $B$ is defined as
    \begin{equation}
        D^{\leftarrow}(\phi^{AB}) \coloneqq I(A:B)_{\phi} - J(A|B)_{\phi}. 
    \end{equation}
    In the above definition, $I(A:B)_{\phi}$ is the quantum mutual information defined as $I(A:B)_{\phi} \coloneqq \entropy(\phi^{A}) + \entropy(\phi^{B}) - \entropy(\phi^{AB})$, and $J(A|B)_{\phi}$ is the classical correlation function~\cite{Henderson2001} of $\phi^{AB}$ defined as
    \begin{equation}
        J(A|B)_{\phi} \coloneqq \entropy(\phi^{A}) - \min_{\{M_k^B\}_k} \entropy(A|\{M_k^B\}_k)_{\phi}
    \end{equation}
    where $\{M_k^B\}_k$ forms a measurement on system $B$; that is, $\sum_k (M_k^B)^\dagger M_k^B = \I^B$ holds, and 
    \begin{equation}
        \entropy(A|\{M_k^B\}_k)_{\phi} \coloneqq \sum_k p_k \entropy(\phi^{A|k})    
    \end{equation}
    is the conditional entropy with measurement $\{M_k^B\}_k$
    with probability $p_k \coloneqq \tr((\I_A \otimes (M_k^B)^\dagger M_k^B)\phi^{AB})$ and  
    post-measurement state $\phi^{A|k} \coloneqq \tr_B((\I^A\otimes M_k^B)\phi^{AB}(\I^A \otimes (M_k^B)^\dagger))/p_k$ for measurement outcome $k$.
    
    In a similar manner, we can define the measurement-based quantum discord for the measurement on $A$ as
    \begin{equation}
        D^{\rightarrow}(\phi^{AB}) \coloneqq I(A:B)_{\phi} - J(B|A)_{\phi}. 
    \end{equation}
\end{definition}

\begin{remark}
    The measurement-based quantum discord is not symmetric; that is, $D^{\leftarrow} \neq D^{\rightarrow}$ in general~\cite{Ollivier2001}. 
\end{remark}
Hereafter, we only consider $D^{\leftarrow}$ as the measurement-based quantum discord. 
The same argument applies to $D^{\rightarrow}$ by interchanging the roles of $A$ and $B$. 
We here show that $D^{\leftarrow}$ is subadditive and asymptotically continuous; in particular, the regularization of $D^{\leftarrow}$ is consistent. 
\begin{proposition}
    The measurement-based quantum discord $D^{\leftarrow}$ is subadditive; that is, for quantum states $\phi^{AB}$ and $\psi^{A'B'}$, it holds that 
    \begin{equation}
        D^{\leftarrow}(\phi^{AB}\otimes\psi^{A'B'}) \leqq D^{\leftarrow}(\phi^{AB}) + D^{\leftarrow}(\psi^{A'B'}). 
    \end{equation}
\end{proposition}
\begin{proof}
    Let $\{M_k^B\}_k$ and $\{N_l^{B'}\}_l$ be optimal measurements achieving $D^{\leftarrow}(\phi^{AB})$ and $D^{\leftarrow}(\psi^{A'B'})$ respectively. 
    Then, it holds that
    \begin{align}
        D^{\leftarrow}(\phi^{AB}\otimes\psi^{A'B'}) 
        &= I(AA':BB')_{\phi\otimes\psi} - J(AA'|BB')_{\phi\otimes\psi}\\
        &= \entropy(\phi^B\otimes\psi^{B'}) - \entropy(\phi^{AB}\otimes\psi^{A'B'}) + \min_{\{O_j^{BB'}\}_j} \entropy(AA'|\{O_j^{BB'}\}_j)_{\phi\otimes\psi}\\
        &= (\entropy(\phi^B) - \entropy(\phi^{AB})) + (\entropy(\psi^{B'}) - \entropy(\psi^{A'B'})) + \min_{\{O_j^{BB'}\}_j} \entropy(AA'|\{O_j^{BB'}\}_j)_{\phi\otimes\psi}. 
    \end{align}
    Here, consider measurement $\{M_k^B\otimes N_l^{B'}\}_{k,l}$. 
    This is a suboptimal mesurement for $ \min_{\{O_j^{BB'}\}_j} \entropy(AA'|\{O_j^{BB'}\}_j')_{\phi\otimes\psi}$. 
    Probability $p_{k,l}$ for measurement outcome $(k,l)$ of this measurement is 
    \begin{align}
        p_{k,l}
        &= \tr((\I^A\otimes M_k^{B}\otimes\I^{A'}\otimes N_l^{B'})(\phi^{AB}\otimes\psi^{A'B'})(\I^A\otimes (M_k^{B})^\dagger\otimes\I^{A'}\otimes (N_l^{B'})^\dagger))\\
        &= \tr((\I^A\otimes M_k^{B})\phi^{AB}(\I^A\otimes (M_k^{B})^\dagger))\times \tr((\I^{A'}\otimes N_l^{B'})\psi^{A'B'}(\I^{A'}\otimes (N_l^{B'})^\dagger))\\
        &\eqqcolon q_{k}\times r_{l}
    \end{align}
    and post-measurement state $\rho^{AA'|k,l}$ for outcome $(k,l)$ is 
    \begin{align}
        \rho^{AA'|k,l} 
        &= \frac{\tr_{BB'}((\I^A\otimes M_k^{B}\otimes\I^{A'}\otimes N_l^{B'})(\phi^{AB}\otimes\psi^{A'B'})(\I^A\otimes (M_k^{B})^\dagger\otimes\I^{A'}\otimes (N_l^{B'})^\dagger))}{p_{k,l}}\\
        &= \frac{\tr_B((\I^A\otimes M_k^{B})\phi^{AB}(\I^A\otimes (M_k^{B})^\dagger))}{q_k} \otimes \frac{\tr_{B'}((\I^{A'}\otimes N_l^{B'})\psi^{A'B'}(\I^{A'}\otimes (N_l^{B'})^\dagger))}{r_l}\\
        &\eqqcolon \phi^{A|k}\otimes\psi^{A'|l}.
    \end{align}
    Then, it holds that 
    \begin{align}
        \min_{\{O_j^{BB'}\}_j} \entropy(AA'|\{O_j^{BB'}\}_j)_{\phi\otimes\psi}
        &\leqq \sum_{k,l} p_{k,l}\entropy(\rho^{AA'|k,l})\\
        &= \sum_{k,l} q_kr_l\entropy(\phi^{A|k}\otimes\psi^{A'|l})\\
        &= \left(\sum_k q_k \entropy(\phi^{A|k})\right) + \left(\sum_l r_l \entropy(\psi^{A'|l})\right).
    \end{align}
    Therefore, it follows that
    \begin{align}
        D^{\leftarrow}(\phi^{AB}\otimes\psi^{A'B'}) 
        &\leqq (\entropy(\phi^B) - \entropy(\phi^{AB})) + (\entropy(\psi^{B'}) - \entropy(\psi^{A'B'})) + \left(\sum_k q_k \entropy(\phi^{A|k})\right) + \left(\sum_l r_l \entropy(\psi^{A'|l})\right)\\
        &=\left(\entropy(\phi^B) - \entropy(\phi^{AB}) + \sum_k q_k \entropy(\phi^{A|k})) \right) + \left(\entropy(\psi^{B'}) - \entropy(\psi^{A'B'}) + \sum_l r_l \entropy(\psi^{A'|l})\right)\\
        \label{eq_MBQD_1}
        &= D^{\leftarrow}(\phi^{AB}) + D^{\leftarrow}(\psi^{A'B'}), 
    \end{align}
    where the last equality~\eqref{eq_MBQD_1} follows from the optimality of $\{M_k^B\}_k$ and $\{N_l^{B'}\}_l$.
\end{proof}

\begin{theorem}
    Let $\mathcal{H}^A$ and $\mathcal{H}^B$ be finite-dimensional quantum systems with dimension $D = \mathrm{dim}(\mathcal{H}^A\otimes\mathcal{H}^B)$.
    The measurement-based quantum discord is asymptotically continuous; that is, for arbitrary states $\phi^{AB}, \psi^{AB}$ on $\mathcal{H}^A\otimes\mathcal{H}^B$, it holds that 
    \begin{equation}
        |D^{\leftarrow}(\phi^{AB}) - D^{\leftarrow}(\psi^{AB})| \leqq 2\|\phi^{AB} - \psi^{AB}\|_1 \log_2D + 4h_2\left(\frac{\|\phi^{AB} - \psi^{AB}\|_1}{2}\right), 
    \end{equation}
    where $h_2$ is the binary entropy function. 
\end{theorem}
\begin{proof}
    By definition, it holds that 
    \begin{align}
        &|D^{\leftarrow}(\phi^{AB}) - D^{\leftarrow}(\psi^{AB})| \\
        &= \left|\left(\entropy(\phi^B) - \entropy(\phi^{AB}) + \min_{\{M_k^{B}\}_k} \entropy(A|\{M_k^{B}\}_k)_\phi \right) - \left(\entropy(\psi^{B}) - \entropy(\psi^{AB}) + \min_{\{N_l^B\}_l} \entropy(A|\{N_l^B\}_l)_\psi\right)\right|\\
        \label{MBQD_AC_0}
        &\leqq |\entropy(\phi^B) - \entropy(\psi^B)| + |\entropy(\phi^{AB}) - \entropy(\psi^{AB})| + \left|\min_{\{M_k^{B}\}_k} \entropy(A|\{M_k^{B}\}_k)_\phi - \min_{\{N_l^B\}_l} \entropy(A|\{N_l^B\}_l)_\psi\right|. 
    \end{align}
Here, since $\|\phi^{B} - \psi^{B}\|_1 \leqq \|\phi^{AB} - \psi^{AB}\|_1$, from Lemma~\ref{lemma_Fannes} (Fannes-Audenaert inequality), it follows that 
\begin{align}
    \label{MBQD_AC_1}
    |\entropy(\phi^B) - \entropy(\psi^B)| &\leqq \frac{\|\phi^{AB} - \psi^{AB}\|}{2}\log_2D + h_2\left(\frac{\|\phi^{AB} - \psi^{AB}\|}{2}\right), \\
    \label{MBQD_AC_2}
    |\entropy(\phi^{AB}) - \entropy(\psi^{AB})| &\leqq \frac{\|\phi^{AB} - \psi^{AB}\|}{2}\log_2D + h_2\left(\frac{\|\phi^{AB} - \psi^{AB}\|}{2}\right).
\end{align}
Now, let us assume that $\min_{\{M_k^{B}\}_k} \entropy(A|\{M_k^{B}\}_k)_\phi \geqq \min_{\{N_l^B\}_l} \entropy(A|\{N_l^B\}_l)_\psi$. 
Otherwise, we can change the roles of $\phi^{AB}$ and $\psi^{AB}$ to apply the same argument. 
Then, let $\{O_j^{B}\}$ be an optimal measurement to achieve $\min_{\{N_l^B\}_l} \entropy(A|\{N_l^B\}_l)_\psi$. 
Define 
\begin{align}
    p_j &\coloneqq \tr((\I^A\otimes O_j^B)\phi^{AB}(\I^A\otimes (O_j^B)^\dagger)),\\
    q_j &\coloneqq \tr((\I^A\otimes O_j^B)\psi^{AB}(\I^A\otimes (O_j^B)^\dagger)),\\
    \phi^{A|j} &\coloneqq \tr_B((\I^A\otimes O_j^B)\phi^{AB}(\I^A\otimes (O_j^B)^\dagger))/p_j,\\
    \psi^{A|j} &\coloneqq \tr_B((\I^A\otimes O_j^B)\psi^{AB}(\I^A\otimes (O_j^B)^\dagger))/q_j
\end{align}
for each label (measurement outcome) $j$. 
Then, it holds that 
\begin{align}
    &\left|\min_{\{M_k^{B}\}_k} \entropy(A|\{M_k^{B}\}_k)_\phi - \min_{\{N_l^B\}_l} \entropy(A|\{N_l^B\}_l)_\psi\right|\\
    &= \min_{\{M_k^{B}\}_k} \entropy(A|\{M_k^{B}\}_k)_\phi - \min_{\{N_l^B\}_l} \entropy(A|\{N_l^B\}_l)_\psi\\
    &\leqq \sum_j p_j\entropy(\phi^{A|j}) - \sum_j q_j\entropy(\psi^{A|j}).
\end{align} 
We here consider quantum channels $\Phi$ and $\Delta$ on system $\mathcal{H}^A \otimes \mathcal{H}^B$ defined as 
\begin{align}
    \Phi(X^{AB}) &\coloneqq \sum_j \tr_B((\I^A\otimes O_j^B)X^{AB}(\I^A\otimes (O_j^B)^\dagger))\otimes\ket{j}\bra{j},\\
    \Delta(X^{AB}) &\coloneqq \sum_j \tr((\I^A\otimes O_j^B)X^{AB}(\I^A\otimes (O_j^B)^\dagger))\ket{j}\bra{j}
\end{align}
for all operators $X^{AB}$ on $\mathcal{H}^A \otimes \mathcal{H}^{B}$. 
Then, it holds that 
\begin{align}
    \sum_j p_j\entropy(\phi^{A|j}) 
    &= \entropy\left(\sum_j p_j\phi^{A|j} \otimes \ket{j}\bra{j}\right) + \sum_j p_j \log_2p_j\\
    &= \entropy(\Phi(\phi^{AB})) - \entropy(\Delta(\phi^{AB})). 
\end{align}
Similarly, we have 
\begin{equation}
    \sum_j p_j\entropy(\psi^{A|j}) = \entropy(\Phi(\psi^{AB})) - \entropy(\Delta(\psi^{AB})). 
\end{equation}
Therefore, it follows that 
\begin{align}
    \label{MBQD_AC_3}
    &\left|\min_{\{M_k^{B}\}_k} \entropy(A|\{M_k^{B}\}_k)_\phi - \min_{\{N_l^B\}_l} \entropy(A|\{N_l^B\}_l)_\psi\right|\\
    &\leqq \left|\entropy(\Phi(\phi^{AB})) - \entropy(\Delta(\phi^{AB})) - \left(\entropy(\Phi(\psi^{AB})) - \entropy(\Delta(\psi^{AB}))\right)\right|\\
    \label{MBQD_AC_4}
    &\leqq \left|\entropy(\Phi(\phi^{AB})) -\entropy(\Phi(\psi^{AB}))\right| + \left|\entropy(\Delta(\phi^{AB})) - \entropy(\Delta(\psi^{AB}))\right|.
\end{align}
Since it holds that $\|\Phi(\phi^{AB}) - \Phi(\psi^{AB})\|_1 \leqq \|\phi^{AB}-\psi^{AB}\|_1$, from Lemma~\ref{lemma_Fannes} (the Fannes-Audenaert inequality), 
we have 
\begin{equation}
    \label{MBQD_AC_5}
    \left|\entropy(\Phi(\phi^{AB})) - \entropy(\Phi(\psi^{AB}))\right|
    \leqq \frac{\|\phi^{AB} - \psi^{AB}\|_1}{2}\log_2D + h_2\left(\frac{\|\phi^{AB}-\psi^{AB}\|_1}{2}\right). 
\end{equation}
In a similar manner, since it holds that $\|\Delta(\phi^{AB}) - \Delta(\psi^{AB})\|_1 \leqq \|\phi^{AB}-\psi^{AB}\|_1$, from Lemma~\ref{lemma_Fannes} (the Fannes-Audenaert inequality), 
we have 
\begin{equation}
    \label{MBQD_AC_6}
    \left|\entropy(\Delta(\phi^{AB})) - \entropy(\Delta(\psi^{AB}))\right| 
    \leqq \frac{\|\phi^{AB} - \psi^{AB}\|_1}{2}\log_2D + h_2\left(\frac{\|\phi^{AB}-\psi^{AB}\|_1}{2}\right). 
\end{equation}
Combining Eqs.~\eqref{MBQD_AC_3}-\eqref{MBQD_AC_6}, we have
\begin{equation}\label{MBQD_AC_7}
    \left|\min_{\{M_k^{B}\}_k} \entropy(A|\{M_k^{B}\}_k)_\phi - \min_{\{N_l^B\}_l} \entropy(A|\{N_l^B\}_l)_\psi\right| \leqq \|\phi^{AB} - \psi^{AB}\|_1\log_2D + 2h_2\left(\frac{\|\phi^{AB}-\psi^{AB}\|_1}{2}\right). 
\end{equation}
Therefore, by Eqs.~\eqref{MBQD_AC_0}, \eqref{MBQD_AC_1}, \eqref{MBQD_AC_2}, and \eqref{MBQD_AC_7}, 
it holds that
\begin{equation}
    |D^{\leftarrow}(\phi^{AB}) - D^{\leftarrow}(\psi^{AB})| \leqq 2\|\phi^{AB} - \psi^{AB}\|_1\log_2D + 4h_2\left(\frac{\|\phi^{AB}-\psi^{AB}\|_1}{2}\right).
\end{equation}
\end{proof}

\section{Non-Markovianity}~\label{SSec_Markov}
In this section, we investigate a consistent resource measure in the resource theory of non-Markovianity. 
We adopt the notation in Ref.~\cite{Ibinson2008}.
A Markov state is defined as a state whose quantum conditional mutual information is zero. 
\begin{definition}
    Let $\mathcal{H}^A$, $\mathcal{H}^B$ and $\mathcal{H}^C$ be finite-dimensional quantum systems. 
    Define the set of quantum Markov state $\mathcal{D}_{\mathrm{Markov}}$ as
    \begin{equation}
        \mathcal{D}_{\mathrm{Markov}} \coloneqq \left\{\psi^{ABC} \in \mathcal{D}\left(\mathcal{H}^A\otimes\mathcal{H}^B\otimes\mathcal{H}^C\right): I(A:C|B)_\psi = 0\right\}. 
    \end{equation}
\end{definition}
\begin{remark}
    $\mathcal{D}_{\mathrm{Markov}}$ is not convex.
\end{remark}
Here, we define the relative entropy of non-Markovianity. 
\begin{definition}
    Let $\mathcal{H}^A$, $\mathcal{H}^B$ and $\mathcal{H}^C$ be finite-dimensional quantum systems. 
    For a quantum state $\phi^{ABC} \in \mathcal{D}\left(\mathcal{H}^A\otimes\mathcal{H}^B\otimes\mathcal{H}^C\right)$, 
    define the relative entropy of non-Markovianity $\Delta$ as
    \begin{equation}
        \Delta(\phi^{ABC}) \coloneqq \min_{\psi^{ABC} \in \mathcal{D}_{\mathrm{Markov}}} \relent(\phi^{ABC}\|\psi^{ABC}).
    \end{equation}
\end{definition}
The following lemma shows a simple characterization of the relative entropy of non-Markovianity.
\begin{lemma}[\cite{Ibinson2008}]~\label{Slemma_decomp_markov}
    Let $\mathcal{H}^A$, $\mathcal{H}^B$ and $\mathcal{H}^C$ be finite-dimensional quantum systems. 
    For all states $\phi^{ABC} \in \mathcal{D}(\mathcal{H}^{A}\otimes\mathcal{H}^{B}\otimes\mathcal{H})$, it holds that
        \begin{equation}
            \Delta(\phi^{ABC}) = \min_{b_j^L, b_j^R} \left[\entropy\left(\bigoplus_j q_j\omega(\phi)^{Ab_j^L}_j\otimes\omega(\phi)^{b_j^RC}_j\right)\right] - \entropy(\phi^{ABC}), 
        \end{equation}
        where 
        the minimum is taken over decompositions $\mathcal{H}^B = \bigoplus_j \mathcal{H}_{b_j^L} \otimes \mathcal{H}_{b_j^R}$, 
        \begin{align}
            &\omega(\phi)^{Ab_j^L}_j \coloneqq \tr_{b_j^RC}\omega(\phi)^{Ab_j^Lb_j^RC}_j,\\
            &\omega(\phi)^{b_j^RC}_j \coloneqq \tr_{Ab_j^L}\omega(\phi)^{Ab_j^Lb_j^RC}_j,
        \end{align}
        and
        \begin{align}
            &q(\phi)_j \coloneqq \tr[(\I^{AC}\otimes P_j)\phi(\I^{AC}\otimes P_j)],\\
            &\omega(\phi)^{Ab_j^Lb_j^RC}_j \coloneqq \frac{(\I^{AC}\otimes P_j)\phi(\I^{AC}\otimes P_j)}{q(\phi)_j},
        \end{align}
        with the projector $P_j$ onto $\mathcal{H}_{b_j^L} \otimes \mathcal{H}_{b_j^R}$.
\end{lemma}
Then, using Lemma~\ref{Slemma_decomp_markov}, we prove the asymptotic continuity of the relative entropy of non-Markovianity. 
\begin{theorem}
    Let $\mathcal{H}^A$, $\mathcal{H}^B$ and $\mathcal{H}^C$ be 
    finite-dimensional quantum systems. 
    Define $D = \mathrm{dim}(\mathcal{H}^A\otimes \mathcal{H}^B \otimes \mathcal{H}^C)$. 
    Let $\phi^{ABC}, \psi^{ABC} \in \mathcal{D}(\mathcal{H}^A\otimes \mathcal{H}^B \otimes \mathcal{H}^C)$ be quantum states. 
    Then, the relative entropy of non-Markovianity $\Delta$ is 
    asymptotically continuous; \textit{i.e.}, 
    if 
    \begin{equation}\label{S_ineq_markov_assumption}
        \frac{\|\phi^{ABC}-\psi^{ABC}\|_1}{2} \leqq \frac{1}{3}, 
    \end{equation}
    it holds that
    \begin{equation}
        |\Delta(\phi^{ABC}) - \Delta(\psi^{ABC})| \leqq 2(\|\phi^{ABC}-\psi^{ABC}\|_1\log_2 D + h_2(\|\phi^{ABC}-\psi^{ABC}\|_1)). 
    \end{equation}
\end{theorem}
    \begin{proof}
        From Lemma \ref{Slemma_decomp_markov}, it holds that
        \begin{align}
            &|\Delta(\phi^{ABC}) - \Delta(\psi^{ABC})| \\
            &\leqq |\entropy(\phi) - \entropy(\psi)| + \left|\entropy\left(\bigoplus_j q(\phi)_j\omega(\phi)^{Ab_j^L}_j\otimes\omega(\phi)^{b_j^RC}_j\right) - \entropy\left(\bigoplus_j q(\psi)_j\omega(\psi)^{Ab_j^L}_j\otimes\omega(\psi)^{b_j^RC}_j\right)\right|,
        \end{align}
        where $\omega$ specifies the optimal decomposition for $\psi$. 
        Here, we assume that
        \begin{equation*} 
        \min_{b_j^L,b_j^R} \entropy\left(\bigoplus_j q(\phi)_j\omega(\phi)^{Ab_j^L}_j\otimes\omega(\phi)^{b_j^RC}_j\right) \geqq \min_{b_j^L,b_j^R} \entropy\left(\bigoplus_j q(\psi)_j\omega(\psi)^{Ab_j^L}_j\otimes\omega(\psi)^{b_j^RC}_j\right).
        \end{equation*}
        Otherwise, we can replace $\phi^{ABC}$ and $\psi^{ABC}$, and apply the same argument. 
        Now, 
        \begin{align}
            &\left\|\bigoplus_j q(\phi)_j\omega(\phi)^{Ab_j^L}_j\otimes\omega(\phi)^{b_j^RC}_j - \bigoplus_j q(\psi)_j\omega(\psi)^{Ab_j^L}_j\otimes\omega(\psi)^{b_j^RC}_j\right\|_1\\
            &\leqq \left\|\bigoplus_j q(\phi)_j\omega(\phi)^{Ab_j^L}_j\otimes\omega(\phi)^{b_j^RC}_j - \bigoplus_j q(\psi)_j\omega(\psi)^{Ab_j^L}_j\otimes\omega(\phi)^{b_j^RC}_j\right\|_1\\
            &\quad + \left\|\bigoplus_j q(\psi)_j\omega(\psi)^{Ab_j^L}_j\otimes\omega(\phi)^{b_j^RC}_j - \bigoplus_j q(\phi)_j\omega(\psi)^{Ab_j^L}_j\otimes\omega(\phi)^{b_j^RC}_j\right\|_1\\
            &\quad + \left\|\bigoplus_j q(\phi)_j\omega(\psi)^{Ab_j^L}_j\otimes\omega(\phi)^{b_j^RC}_j - \bigoplus_j q(\psi)_j\omega(\psi)^{Ab_j^L}_j\otimes\omega(\psi)^{b_j^RC}_j\right\|_1\\
            &= \sum_j \left\| q(\phi)_j\omega(\phi)^{Ab_j^L}_j -  q(\psi)_j\omega(\psi)^{Ab_j^L}_j\right\|_1
            + \sum_j \left| q(\psi)_j -  q(\phi)_j\right|_1
            + \sum_j \left\|q(\phi)_j\omega(\phi)^{b_j^RC}_j - q(\psi)_j\omega(\psi)^{b_j^RC}_j\right\|_1\\
            &\leqq 3\sum_j\left\|(\I^{AC}\otimes P_j)(\phi^{ABC} - \psi^{ABC})(\I^{AC}\otimes P_j)\right\|_1\\
            &\leqq 3\left\|\bigoplus_j (\I^{AC}\otimes P_j)(\phi^{ABC} - \psi^{ABC})(\I^{AC}\otimes P_j)\right\|_1\\
            &\leqq 3\left\|\phi^{ABC} - \psi^{ABC}\right\|_1.
        \end{align}
        Therefore, if $\|\phi^{ABC}-\psi^{ABC}\|_1/2 \leqq 1/3$, 
        \begin{equation}\label{S_ineq_markov1}
            \frac{1}{2}\left\|\bigoplus_j q(\phi)_j\omega(\phi)^{Ab_j^L}_j\otimes\omega(\phi)^{b_j^RC}_j - \bigoplus_j q(\psi)_j\omega(\psi)^{Ab_j^L}_j\otimes\omega(\psi)^{b_j^RC}_j\right\|_1\leqq \frac{3}{2}\|\phi^{ABC}-\psi^{ABC}\|_1 \leqq 1, 
        \end{equation}
        and from Lemma~\ref{lemma_Fannes} (Fannes-Audenaert inequality), we have
        \begin{align}
            &|\Delta(\phi^{ABC}) - \Delta(\psi^{ABC})|\\
            &\leqq \frac{\|\phi^{ABC}-\psi^{ABC}\|_1}{2}\log_2D + h_2\left(\frac{\|\phi^{ABC}-\psi^{ABC}\|_1}{2}\right) + \frac{3\|\phi^{ABC}-\psi^{ABC}\|_1}{2}\log_2D + h_2\left(\frac{3\|\phi^{ABC}-\psi^{ABC}\|_1}{2}\right)\\
            &\leqq 2(\|\phi^{ABC}-\psi^{ABC}\|_1\log_2D + h_2(\|\phi^{ABC}-\psi^{ABC}\|_1)), 
        \end{align}
        where the last inequality follows from the concavity of $h_2$.
    \end{proof}

    \begin{remark}
        One may remove assumption~\eqref{S_ineq_markov_assumption} in the following way. 
        If we do not assume Eq.~\eqref{S_ineq_markov_assumption}, instead of Eq.~\eqref{S_ineq_markov1}, we have 
        \begin{equation}
            \frac{1}{2}\left\|\bigoplus_j q(\phi)_j\omega(\phi)^{Ab_j^L}_j\otimes\omega(\phi)^{b_j^RC}_j - \bigoplus_j q(\psi)_j\omega(\psi)^{Ab_j^L}_j\otimes\omega(\psi)^{b_j^RC}_j\right\|_1\leqq \min\left\{\frac{3}{2}\|\phi^{ABC}-\psi^{ABC}\|_1, 1\right\}.
        \end{equation}
        This leads to
        \begin{equation}
            |\Delta(\phi^{ABC}) - \Delta(\psi^{ABC})| 
            \leqq \left\{
            \begin{array}{cc}
            2(\|\phi^{ABC}-\psi^{ABC}\|_1\log_2D + h_2(\|\phi^{ABC}-\psi^{ABC}\|_1)) & \left(\|\phi^{ABC}-\psi^{ABC}\|_1 \leqq \frac{2}{3}\right),\\
            \frac{\|\phi^{ABC}-\psi^{ABC}\|_1}{2}\log_2D + h_2\left(\frac{\|\phi^{ABC}-\psi^{ABC}\|_1}{2}\right) + \log_2D & \left(\|\phi^{ABC}-\psi^{ABC}\|_1 > \frac{2}{3}\right).
            \end{array}
            \right.
        \end{equation}
        However, we here adopt assumption~\eqref{S_ineq_markov_assumption} for simplicity of presentation, which does not affect the argument of asymptotic continuity. 
    \end{remark}

    \section{Non-Gaussianity}~\label{SSec_NG}
    In this section, we investigate a consistent resource measure in the resource theory of non-Gaussianity. 
    In the following, we consider an infinite-dimensional Hilbert space. 
    Here, we recall the basic concepts of Gaussian quantum information~\cite{Braunstein2005,Weedbrook2012}.
    Consider an $N$-mode bosonic system and define real quadratic field operators $\hat{x} \coloneqq (\hat{q}_1,\hat{p}_1,\ldots,\hat{q}_N,\hat{p}_N)$. 
    A quantum state $\phi$ is described by its Wigner characteristic function $\chi(\xi,\phi) = \tr[\phi \hat{D}(\xi)]$.
    Here, $\xi = (\xi_1,\ldots,\xi_{2N}) \in \mathbb{R}^{2N}$ is a real-valued vector, and $\hat{D}(\xi) \coloneqq \exp(i \hat{x}^\mathrm{T}\Omega\xi)$ is the Weyl displacement operator with $\Omega \coloneqq i\bigoplus_{k=1}^N Y$, where $Y$ is the Pauli-$Y$ matrix.
    Then, a quantum state $\psi$ is a Gaussian state if its characteristic function $\chi(\xi,\psi)$ is expressed as a Gaussian function whose form is determined only by the mean and covariance matrix of $\hat{x}$ with respect to $\psi$
    \begin{equation}
        \chi(\xi,\psi) = \exp\left(-\frac{1}{2}\xi^{\mathrm{T}}(\Omega\Lambda\Omega^{\mathrm{T}})\xi - i(\Omega\bar{x})^{\mathrm{T}}\xi\right), 
    \end{equation}
    where $\bar{x}$ is the mean of $\hat{x}$ with respect to the state $\psi$ and $\Lambda$ is its covariance matrix defined as $\bar{x}_{i} \coloneqq \braket{\hat{x}_i} \coloneqq \tr(\hat{x}_i\phi)$ and $\Lambda{ij} \coloneqq \braket{\hat{x}_i\hat{x}_j + \hat{x}_j\hat{x}_i}/2 - \braket{\hat{x}_i}\braket{\hat{x}_j}$ for all $i,j \in \{1,2,\ldots,2n\}$.

    \subsection{Nonconvex}~\label{SSubsec_nonconvex}
    We first provide the definition of the relative entropy of non-Gaussianity. 
    Note that since the sum of two Gaussian functions is not necessarily Gaussian, the set of Gaussian states $\mathcal{G}$ is not convex. 
    \begin{definition}  
           Let $\mathcal{H}$ be an infinite-dimensional Hilbert space. 
           Let $\phi \in \mathcal{D}(\mathcal{H})$ be a state. 
           Define the relative entropy of non-Gaussianity $\delta$ as
           \begin{equation}
            \delta[\phi] \coloneqq \min_{\psi \in \mathcal{G}} \relent(\phi\|\psi),
           \end{equation}
           where $\mathcal{G}$ is the set of the Gaussian states. 
    \end{definition}
    It is known that the relative entropy of non-Gaussianity can be written as
    \begin{equation}
        \delta[\phi] \coloneqq \entropy(\phi_{\textup{G}}) - \entropy(\phi), 
    \end{equation}
    where $\phi_{\textup{G}} \in \mathcal{G}$ is the Gaussification of $\phi$, which has the same mean and covariance matrix as $\phi$~\cite{Geroni2008}. 
    
    In the following theorem, we show that the relative entropy of non-Gaussianity is not continuous even under a reasonable energy constraint by a counterexample. 
    \begin{theorem}
        Let $\mathcal{H}$ be a single-mode space. 
        The relative entropy of $\delta[\cdot]$ is not continuous; that is, 
        there exists $\epsilon>0$ such that for any $0<\epsilon'\leq1$, 
        there exist states $\rho, \sigma \in \mathcal{D}(\mathcal{H})$ 
        with energy constraint $\tr(H\rho), \tr(H\sigma) \leqq E $ for a given Hamiltonian $H = \hbar\omega a^\dagger a$ and some positive real number $E$ 
        such that
        \begin{align}
            &\frac{1}{2}\|\rho - \sigma\|_1 \leqq \epsilon',\\
            &|\delta[\rho] - \delta[\sigma]| > \epsilon. 
        \end{align} 
    \end{theorem}
    \begin{proof}
        Let $0<\alpha\leqq 1$ be a real number. 
        Define an integer
        \begin{equation}
            m \coloneqq \left\lfloor\frac{E}{\alpha\hbar\omega} \right\rfloor,
        \end{equation}
        where $\lfloor{}\cdots{}\rfloor$ is the floor function.
        Then, define $\rho$ and $\omega$ as
        \begin{align}
            &\rho \coloneqq \alpha\ket{m}\bra{m} + \left(1-\alpha\right)\ket{0}\bra{0},\\
            &\sigma \coloneqq \ket{0}\bra{0}. 
        \end{align}
        Then, it holds that
        \begin{equation}
            \frac{1}{2}\|\rho - \sigma \|_1 = \alpha. 
        \end{equation}
        Moreover, it holds that
        \begin{align}
            &\tr(H\rho) = \alpha\hbar\omega m \leqq E,\\
            &\tr(H\sigma) = 0 \leqq E. 
        \end{align}
        Since the Gaussification $\omega_{G} \in \mathcal{G}$ of a Fock-diagonal state
        \begin{equation}
            \omega = \sum_{l=0}^\infty p_l\ket{l}\bra{l}
        \end{equation}
        is given as
        \begin{equation}
            \omega_{\textup{G}} = \sum_{l=0}^\infty \frac{n^l}{(n+1)^{l+1}}\ket{l}\bra{l}
        \end{equation}
        with $n = \tr{a^\dagger a \omega} = \sum_{l=0}^\infty lp_l$~\cite{Marian2013a,Marian2013b}, it holds that
        \begin{align}
            \delta(\omega)
            &= \entropy(\omega_{\textup{G}}) - \entropy(\omega)\\
            &= -\sum_{l=0}^\infty \frac{n^l}{(n+1)^{l+1}}\log_2 \frac{n^l}{(n+1)^{l+1}} - \entropy(\omega)\\
            &= (n+1)\log_2(n+1) - n\log_2n -\entropy(\omega)\\
            &= (n+1)\log_2(n+1) - n\log_2n + \sum_{l=0}^\infty p_l\log_2p_l. 
        \end{align}
        Therefore, it holds that
        \begin{equation}
            |\delta[\rho] - \delta[\sigma]| = \left(\alpha m+1\right) \log_2\left(\alpha m+1\right) - \alpha m\log_2(\alpha m) - h_2\left(\alpha\right)
        \end{equation}
        Define a function $f(x)\coloneqq (x+1)\log_2(x+1) - x\log_2x$ so that 
        we can write
        \begin{equation}
            |\delta[\rho] - \delta[\sigma]| = f(\alpha m) - h_2(\alpha). 
        \end{equation}
        Since $f'(x) = \log_2(1+1/x) > 0$, $f$ is an increasing function. 
        In addition, it holds that
        \begin{equation}
            \alpha m = \alpha\left\lfloor \frac{E}{\alpha \hbar\omega} \right\rfloor > \frac{E}{\hbar \omega} -\alpha. 
        \end{equation}
        Therefore, it holds that
        \begin{equation}
            |\delta[\rho] - \delta[\sigma]| > f\left(\frac{E}{\hbar \omega} -\alpha\right) - h_2(\alpha) \eqqcolon g(\alpha). 
        \end{equation}
        Since $f$ is increasing, $g(0) = f(E/\hbar\omega) > f(0) = 0$. 
        Therefore, since $g$ is continuous, there must be a real number $0 < \alpha_0\leqq 1/2$ such that 
        $g(\alpha_0) > 0$. 

        Fix such an $\alpha_0$. 
        We prove that we can take $\epsilon \coloneqq g(\alpha_0)$.
        For a given $0 < \epsilon' \leqq 1$, define two states $\rho$ and $\sigma$ as
        \begin{align}
            &\rho \coloneqq 
            \left \{
            \begin{array}{ll}
            \epsilon'\ket{m}\bra{m} + \left(1-\epsilon'\right)\ket{0}\bra{0} & (0<\epsilon'\leqq \alpha_0)\\
            \alpha_0\ket{m}\bra{m} + \left(1-\alpha_0\right)\ket{0}\bra{0} & (\epsilon' \geqq \alpha_0)
            \end{array}
            \right.\\
            &\sigma \coloneqq \ket{0}\bra{0}. 
        \end{align}
        Then, it holds that
        \begin{equation}
            \frac{1}{2}\left\|\rho - \sigma\right\|_1 \leqq \epsilon'. 
        \end{equation} 
        Since $f$ is increasing and since $h_2(x)$ is increasing in $0<x\leqq 1/2$, 
        $g(\alpha) = f\left(E/\hbar\omega -\alpha\right) - h_2(\alpha)$ is decreasing in $0<\alpha\leqq 1/2$.
        Therefore, it holds that
        \begin{equation}
            |\delta[\rho] - \delta[\sigma]| > f\left(\frac{E}{\hbar \omega} -\alpha_0\right) - h_2(\alpha_0) = g(\alpha_0) = \epsilon, 
        \end{equation}
        which implies that $\delta$ is discontinuous. 
    \end{proof}

\subsection{Convex}~\label{SSubsec_convex}
Here, we adopt another framework of the QRT of non-Gaussianity where we consider the convex hull of the set of Gaussian states as the set of free states~\cite{Takagi2018,Albarelli2018}. 
We consider the following modified version of relative entropy of non-Gaussianity, 
which we call the relative entropy of convex non-Gaussianity. 
\begin{definition}  
    Let $\mathcal{H}$ be an infinite-dimensional Hilbert space. 
    Let $\phi \in \mathcal{D}(\mathcal{H})$ be a state. 
    Define the relative entropy of non-Gaussianity $\delta_{\mathrm{conv}}$ as
    \begin{equation}
        \delta_{\textup{conv}}[\phi] \coloneqq \inf_{\psi \in \mathrm{conv}(\mathcal{G})} \relent(\phi\|\psi), 
    \end{equation}
    where the infimum can be replaced with the minimum under the energy constraint that we discuss in the following.
\end{definition}
Here, we prove that $\delta_{\textup{conv}}$ is asymptotically continuous if a given Hamiltonian satisfies a certain condition (See Eq.~\eqref{eq:Gaussian_condition}). 
Hereafter, we suppose that the smallest eigenvalue of the Hamiltonian $H$ is fixed to zero. 
To discuss asymptotic continuity under energy constraint, we consider the following subset of density operators.
\begin{definition}
    Let $H$ be a Hamiltonian of the system $\mathcal{H}$. 
    Let $E$ be a positive number corresponding to an upper bound of the energy of the system such that $E \geqq 0$.
    Then, we define $\mathcal{D}_{H,E}(\mathcal{H}) \subset \mathcal{D}(\mathcal{H})$ as a subset of the set of density operators that contains all density operators $\rho$ satisfying $\tr(H\rho) \leqq E$; 
    that is, $\mathcal{D}_{H,E}(\mathcal{H}) \coloneqq \{\rho \in \mathcal{D}(\mathcal{H}): \tr(H\rho) \leqq E\}$. 
\end{definition}
The core of the proof is the uniform continuity of $\delta_{\textup{conv}}$, which is shown in Ref.~\cite{Shirokov2018} in a general form. 
One of the main components of the proof is the following Gibbs state defined with a Hamiltonian $H$ and an energy $E$
\begin{equation}
    \gamma(H,E) \coloneqq \frac{\mathrm{e}^{-\beta(H,E)H}}{Z(H,E)}, 
\end{equation}
where $Z(H,E)\coloneqq \tr(\mathrm{e}^{-\beta(H,E)H})$ and $\beta(H,E)$ is the solusion of $\tr[\mathrm{e}^{-\beta H}(H-E)] = 0$. 
To show asymptotic continuity of $\delta_{\mathrm{conv}}$, we assume a certain condition given in the following lemma. 
\begin{lemma}[\cite{Shirokov2018}]
    Let $H$ be a Hamiltonian of the system $\mathcal{H}$. 
    Let $E$ be a positive number corresponding to an upper bound of the energy of the system such that $E \geqq 0$.
    Then, $H$ satisfies the condition 
    \begin{equation}\label{eq:Gaussian_condition}
        \lim_{\lambda \to 0} \left[\tr(\mathrm{e}^{-\lambda H})\right]^\lambda = 0, 
    \end{equation}
   if and only if it holds that $H(\gamma(H,E)) = o(\sqrt{E})$. 
\end{lemma}
In the following proposition, we reformulate the statement of Proposition 3 in Ref.~\cite{Shirokov2018} for $\delta_{\textup{conv}}$. 
\begin{proposition}\label{non-Gaussianity_uniform_continuity}
    Let $H$ be a Hamiltonian of the system $\mathcal{H}$. 
    Let $E$ be a positive constant corresponding to an upper bound of the energy of the system.
    Suppose that the condition~\eqref{eq:Gaussian_condition} is satisfied. 
    Then, the relative entropy of convex non-Gaussianity $\delta_{\textup{conv}}$ is uniformly continuous on $\mathcal{D}_{H,E}(\mathcal{H})$; that is, 
    it holds that
    \begin{equation}
        |\delta_{\textup{conv}}[\phi] - \delta_{\textup{conv}}[\psi]| \leqq \sqrt{2\epsilon}\entropy\left(\gamma(H,E/\epsilon)\right) + g_2(\sqrt{2\epsilon})
    \end{equation}
    for any $\phi,\psi \in \mathcal{D}_{H,E}(\mathcal{H})$ such that $\|\phi-\psi\|_1 \leqq \epsilon \leqq 1/2$, where $g_2(x) \coloneqq (x+1)\log_2(x+1) - x\log_2x = (x+1)h_2(x/(x+1))$ defined in $x\in[0,1]$.
\end{proposition}
Using this proposition, we prove asymptotic continuity in the following theorem,
\begin{theorem}~\label{non-Gaussianity_asymptotic_continuity}
    Let $H$ be a Hamiltonian satisfying the condition~\eqref{eq:Gaussian_condition}. 
    Let $E$ be a positive constant. 
    Then, the relative entropy of convex non-Gaussianity $\delta_{\textup{conv}}$ is asymptotically continuous on $\mathcal{D}_{H,E}(\mathcal{H})$; that is, 
    for any sequences $\{\phi_n \in \mathcal{D}_{H_n,nE}(\mathcal{H})\}_n$ and $\{\psi_n \in \mathcal{D}_{H_n,nE}(\mathcal{H}) \}_n$ such that $\lim_{n\to\infty} \|\phi_n - \psi_n\|_1 \to 0$ 
    where the Hamiltonian $H_n = H\otimes \I \otimes \cdots \otimes \I + \I \otimes H \otimes \I \otimes \cdots \otimes \I + \cdots + \I\otimes \cdots \otimes \I\otimes H$ satisfies the condition~\eqref{eq:Gaussian_condition}, 
    it holds that
    \begin{equation}
        \lim_{n\to\infty} \frac{|\delta_{\textup{conv}}[\phi_n] - \delta_{\textup{conv}}[\psi_n]|}{n} = 0. 
    \end{equation}
\end{theorem}
\begin{proof}
    Using Proposition~\ref{non-Gaussianity_uniform_continuity}, it holds that
    \begin{equation}
        |\delta_{\textup{conv}}[\phi_n] - \delta_{\textup{conv}}[\psi_n]| \leqq \sqrt{2\epsilon}\entropy\left(\gamma(H_n,nE/\epsilon)\right) + g_2(\sqrt{2\epsilon}). 
    \end{equation}
    Since $\entropy\left(\gamma(H_n,nE/\epsilon)\right) = \entropy\left(\gamma(H,E/\epsilon)^{\otimes n}\right) = n\entropy\left(\gamma(H,E/\epsilon)\right)$, we have
    \begin{equation}
        \frac{|\delta_{\textup{conv}}[\phi_n] - \delta_{\textup{conv}}[\psi_n]|}{n} \leqq \sqrt{2\epsilon}\entropy\left(\gamma(H,E/\epsilon)\right) + \frac{1}{n}g_2(\sqrt{2\epsilon}). 
    \end{equation}   
\end{proof}
\begin{remark}[Tighter bound for asymptotic continuity]
In fact, a tighter (and almost optimal) uniform-continuity bound was proposed in Ref.~\cite{Shirokov2020} (Theorem 1), 
and we may use this bound instead of Proposition~\ref{non-Gaussianity_uniform_continuity}. 
However, since the assumptions needed to derive the bound are more complicated for technical reasons, 
we adopt the bound in Proposition~\ref{non-Gaussianity_uniform_continuity} in this paper. 
One may rewrite the statements of Proposition~\ref{non-Gaussianity_uniform_continuity} and Theorem~\ref{non-Gaussianity_asymptotic_continuity} using the tighter bound 
with this modification of the assumptions, but we here use the bound in Ref.~\cite{Shirokov2020} for simplicity of analysis.
\end{remark}
Here, we consider the energy-constrained version of the state conversion 
and show that the regularized relative entropy of convex non-Gaussianity 
\begin{equation}
    \delta_{\textup{conv}}^\infty[\phi] \coloneqq \lim_{n\to\infty}\frac{\delta_{\textup{conv}}[\phi^{\otimes n}]}{n}
\end{equation}
becomes a consistent resource measure under the energy-constrained state conversion. 
\begin{definition}
    Let $H_1$ and $H_2$ be Hamiltonians of the system $\mathcal{H}_1$ and $\mathcal{H}_2$ respectively. 
    Let $E_1$ and $E_2$ be some positive constants. 
    Let $\mathcal{O}^{(\mathrm{EC})}(\mathcal{H}_1^{(H_1,E_1)}\to\mathcal{H}_2^{(H_2,E_2)})$ denote the energy-constrained subset of the free operations; that is,
    an opeartion $\mathcal{N} \in \mathcal{O}^{(\mathrm{EC})}(\mathcal{H}_1^{(H_1,E_1)}\to\mathcal{H}_2^{(H_2,E_2)})$ is a map from $\mathcal{D}_{H_1,E_1}(\mathcal{H}_1)$ to $\mathcal{D}_{H_2,E_2}(\mathcal{H}_2)$.   
    
    Then, consider a quantum system $\mathcal{H}$ with a Hamiltonian $H$. Let $E$ be a positive constant corresponding to an upper bound of the energy of the system. 
    For the $m$-fold quantum system $\mathcal{H}^{\otimes m}$, we adopt the following Hamiltonian $H_m \coloneqq H\otimes \I \otimes \cdots \otimes \I + \I \otimes H \otimes \I \otimes \cdots \otimes \I + \cdots + \I \otimes \cdots \otimes \I \otimes H$.
    Define the set of achievable rates under the energy-constrained free operations as
    \begin{equation}
        \label{eq:conversion_rate_set}
      \begin{aligned}
      &\mathcal{R}^{(H,E)}\left(\phi\to\psi\right)\\
      &\coloneqq
      \Big\{r\geqq 0:\exists\left(\mathcal{N}_n\in\mathcal{O}^{(\mathrm{CE})}\left((\mathcal{H}^{\otimes n})^{(H_n,nE)}\to(\mathcal{H}^{\otimes \left\lceil rn\right\rceil})^{(H_{\lceil rn \rceil},\lceil rn \rceil E)}\right):n\in\mathbb{N}\right),
      \liminf_{n\to\infty}\left\|\mathcal{N}_n\left(\phi^{\otimes n}\right)-\psi^{\otimes \left\lceil rn\right\rceil}\right\|_1 = 0\Big\},
      \end{aligned}
    \end{equation}
    where $\phi^{\otimes 0}\coloneqq 1$. 
    With respect to this achievable rate set, an asymptotic state conversion rate $r^{H,E}\left(\phi\to\psi\right)$ is defined as
    \begin{equation}
      \label{eq:conversion_rate}
      r^{(H,E)}\left(\phi\to\psi\right)\coloneqq\sup\mathcal{R}^{(H,E)}\left(\phi\to\psi\right).
    \end{equation}
\end{definition}
\begin{theorem}
    Let $H$ be a Hamiltonian satisfying the condition~\eqref{eq:Gaussian_condition}. 
    Let $E$ be a positive constant. 
    For states $\phi,\psi \in \mathcal{D}_{H,E}(\mathcal{H})$, it holds that 
    \begin{equation}
        \delta_{\textup{conv}}^{\infty}[\psi]r^{(H,E)}\left(\phi\to\psi\right) \leqq \delta_{\textup{conv}}^{\infty}[\phi].
    \end{equation}
\end{theorem}
\begin{proof}
    Fix a positive number $\epsilon > 0$. 
    For an arbitrary positive number $\epsilon'>0$, take $r \coloneqq r^{(H,E)}\left(\phi\to\psi\right)-\epsilon'$.
    There exists a sequence of the energy-constrained free operations $\left(\mathcal{N}_n:n\in\mathbb{N}\right)$ such that
    \begin{equation}
        \left\|\mathcal{N}_n(\phi^{\otimes n}) - \psi^{\otimes \lceil rn \rceil}\right\|_1 < \epsilon. 
    \end{equation}
    holds for an infinitely large number of $n$. 
    Therefore, it holds that
    \begin{equation}
        \begin{aligned}
            \frac{\delta_{\textup{conv}}[\phi^{\otimes n}]}{n}
            &\geqq \frac{\delta_{\textup{conv}}[\mathcal{N}(\phi^{\otimes n})]}{n}\\
            &\geqq \frac{\delta_{\textup{conv}}[\psi^{\otimes{\lceil rn \rceil}}]}{n} - \sqrt{2\epsilon}\entropy\left(\gamma(H,E/\epsilon)\right) - \frac{1}{n}g_2(\sqrt{2\epsilon})\\
            &= \frac{\lceil rn \rceil}{n} \frac{\delta_{\textup{conv}}[\psi^{\otimes{\lceil rn \rceil}}]}{\lceil rn \rceil} - \sqrt{2\epsilon}\entropy\left(\gamma(H,E/\epsilon)\right) - \frac{1}{n}g_2(\sqrt{2\epsilon})\\
            &\geqq r\frac{\delta_{\textup{conv}}[\psi^{\otimes{\lceil rn \rceil}}]}{\lceil rn \rceil} - \sqrt{2\epsilon}\entropy\left(\gamma(H,E/\epsilon)\right) - \frac{1}{n}g_2(\sqrt{2\epsilon}).
        \end{aligned}
    \end{equation}
    Since we can take arbitrarily small $\epsilon$ and $\epsilon'$, it holds that
    \begin{equation}
        \frac{\delta_{\textup{conv}}[\phi^{\otimes n}]}{n} \geqq  r^{(H,E)}\left(\phi\to\psi\right)\frac{\delta_{\textup{conv}}[\psi^{\otimes{\lceil rn \rceil}}]}{\lceil rn \rceil}
    \end{equation}
    for an infinitely large subset of $n$. Taking limit $n \to \infty$, we have
    \begin{equation}
        \delta_{\textup{conv}}^{\infty}[\phi] \geqq  r^{(H,E)}\left(\phi\to\psi\right)\delta_{\textup{conv}}^{\infty}[\psi]. 
    \end{equation}
\end{proof}
\end{document}